\pdfoutput=1

\documentclass[letterpaper,twocolumn,10pt]{article}
\usepackage{usenix,epsfig,endnotes}

\usepackage[T1]{fontenc}
\usepackage{times}
\newcommand{\sys}{\text{Lemonshark}\xspace}

\usepackage{xspace}
\usepackage{graphicx} 
\usepackage{amsmath}
\usepackage[capitalise]{cleveref}
\usepackage{amsthm}

\usepackage{hyperref}

\usepackage{xcolor}
\usepackage{algorithm}
\usepackage{algpseudocode}
\usepackage{float}
\usepackage{subcaption}
\usepackage{titling}

\usepackage{thmtools}
\usepackage[para]{footmisc}
\usepackage{bm}
\usepackage{amsbsy}

\usepackage{enumitem}
\usepackage{titlesec}
\titlespacing\section{0pt}{*0.5}{*0.2}        
\titlespacing\subsection{0pt}{*0.3}{*0.2}       
\titlespacing\subsubsection{0pt}{*0.2}{*0.0}
\titlespacing\paragraph{0pt}{*0.5}{*0.5}

\crefformat{section}{\S#2#1#3} 
\crefformat{subsection}{\S#2#1#3}
\crefformat{subsubsection}{\S#2#1#3}

\declaretheoremstyle[
    spaceabove=1pt,
    spacebelow=2pt,
    headfont=\normalfont\bfseries,
    bodyfont=\itshape, 
    headpunct=.,
]{compacttheorem}

\declaretheorem[style=compacttheorem, numberwithin=section, name=Definition]{definition}
\declaretheorem[style=compacttheorem, numberwithin=section, name=Lemma]{lemma}
\declaretheorem[style=compacttheorem, numberwithin=section, name=Proposition]{proposition}

\def\Snospace~{\S{}}

\newcommand{\mich}[1]{}
\newcommand{\al}[1]{}
\newcommand{\lijl}[1]{}
\newcommand{\yihan}[1]{}
\newcommand{\xiangl}[1]{}

\usepackage{hyperref}
\usepackage{xurl}  
\usepackage[font=small]{caption}

\begin{document}

\date{}

\setlength{\droptitle}{-1em}
\title{\sys: Asynchronous DAG-BFT With Early Finality}


\author{
{\rm Michael Yiqing Hu}\\
National University of Singapore\\
\and
{\rm Alvin Hong Yao Yan}\\
National University of Singapore\\
\and
{\rm Yihan Yang}\\
National University of Singapore\\
\and
{\rm Xiang Liu}\\
National University of Singapore\\
\and
{\rm Jialin Li}\\
National University of Singapore\\
}

\maketitle
\thispagestyle{empty}


\begin{abstract}

DAG-Rider popularized a new paradigm of DAG-BFT protocols, separating dissemination from consensus: all nodes disseminate transactions as blocks that reference previously known blocks, while consensus is reached by electing certain blocks as leaders. This design yields high throughput but confers optimal latency only to leader blocks; non-leader blocks cannot be committed independently.

We present Lemonshark, an asynchronous DAG-BFT protocol that reinterprets the DAG at a transactional level and identifies conditions where commitment is sufficient---but not necessary---for safe results, enabling nodes to finalize transactions before official commitment, without compromising correctness. Compared to the state-of-the-art asynchronous BFT protocol, Lemonshark reduces latency by up to 65\%.

\end{abstract}

\section{Introduction}

Modern distributed applications require consensus protocols that can achieve both high throughput and low latency in adversarial environments.
From cryptocurrency networks\cite{digitalEuro,example1,example2} processing millions of transactions daily to geo-distributed databases\cite{ Petersen1996BayouRD,Lloyd2011DontSF,shuai2016,Shen2025MakoSD} serving global user bases, these systems must maintain consistency despite potential node failures and malicious behavior.

Byzantine Fault-Tolerant (BFT) consensus protocols~\cite{Bracha1987AsynchronousBA,Dolev1982AnEA,Cachin2001SecureAE} address this challenge by solving the Byzantine Atomic Broadcast (BAB) problem~\cite{Correia2006FromCT}:  ensuring all honest processes agree on a total order of messages (e.g., transactions) despite Byzantine faults, with both safety and liveness guarantees.

While BFT protocols solve this fundamental problem, classical Asynchronous BFT approaches\cite{Lu2020DumboMVBAOM,honeybadger,Duan2018BEATAB,Liu2020EPICEA} utilize paradigms\cite{Cachin2001SecureAE,ben-or,bkr} that have been perceived as either excessively complex or prohibitively costly for practical deployment.
While recent advances have substantially improved performance~\cite{Abraham2019AsymptoticallyOV,Lu2020DumboMVBAOM}, a novel class of DAG-based protocols \cite{Gagol2019AlephEA,hashgraph}  has emerged, introducing an innovative paradigm that promises optimal round and amortized communication complexity~\cite{bullshark,dag-rider,narwhal}.


DAG-based BFT protocols~\cite{bullshark,dag-rider,narwhal} follow the principle of decoupled data dissemination and consensus to improve throughput and expected latency.
These protocols enable every node to propose transactions in the form of blocks that reference previously known blocks, collectively forming a Directed Acyclic Graph (DAG) across rounds. This parallelization is the key to their increased throughput. 
While the DAG imposes only a partial order, a total order can be derived by selecting specific blocks as leaders. 

DAG-BFT protocols achieve consensus by electing \textit{leader blocks} using a Global Perfect Coin abstraction~\cite{Boneh2001ShortSF,Libert2014BornAR,Shoup2000PracticalTS}. This mechanism significantly improves expected latencies compared to traditional approaches~\cite{Lu2020DumboMVBAOM,honeybadger,Duan2018BEATAB,Liu2020EPICEA}.
Once confirmed, leader blocks and their causal histories (all blocks reachable through directed paths) are \textbf{committed}, then ordered deterministically and executed to produce \textbf{finalized} outcomes for the constituent transactions of blocks.

Prevailing designs couple \textit{finalized} outcomes with block \textit{commitment}.
The leader election process fundamentally governs block commitment.
\textit{Elected leaders may achieve optimal latency} (committed after the next round), while non-leader blocks require inclusion in the causal history of a subsequently elected leader to be committed. This asymmetry creates a latency disparity: even under perfect network conditions without faults, non-leader blocks experience 1-2 additional rounds of delay compared to leader blocks.

Moreover, leader elections typically yield a single leader and occur infrequently, thereby exacerbating latencies for the majority of blocks.
The current state-of-the-art asynchronous DAG-BFT protocol, Bullshark~\cite{bullshark}, has only partially addressed this issue by enabling more frequent leader elections during periods of relative progress.

Rather than accelerating leader election, \sys introduces a theoretical shift:
\begin{quote}
\vspace{-7pt}
    \textit{We may yield \textbf{finalized} outcomes for transactions in a non-leader block \textbf{before} it is committed.}
\vspace{-7pt}
\end{quote}

We term this capability \textbf{Early Finality} --- which decouples finalization from commitment.
\sys demonstrates that commitment via inclusion in a committed leader block's causal history represents a sufficient, but not necessary, condition for obtaining finalized outcomes when it comes to transactions in asynchronous DAG-BFT protocols.

Naively attempting to derive \textit{finalized} outcomes for transactions contained within non-leader blocks before their \textit{commitment} is generally unsafe. 
In asynchronous, adversarial environments, a node's local view of the global DAG may not be complete. 
As such, there may exist blocks unknown to the node---whose transactions conflict with or otherwise influence outcomes of those in its local view. 
Even if the node's local view includes all the blocks, the inclusion of a given block in the causal history of a subsequently elected leader cannot always be guaranteed.

\sys addresses these safety challenges through a key insight: knowing the ordering of conflicting blocks within the eventual committed leader block’s causal history enables the determination of safe, finalized outcomes even before the leader block is proposed.
We implement this through key-space partitioning and deterministic causal history ordering.

\sys partitions the key-space across blocks within each round, ensuring that each block operates on a distinct shard and eliminating most intra-round conflicts.
This constraint, combined with predefined ordering rules, enables nodes to safely derive final transaction outcomes by observing only conflicting actions from previous rounds --- a practically achievable requirement even under network asynchrony and Byzantine adversaries.
While key-space sharding introduces complexity for cross-shard transactions, \sys achieves early finality by re-interpreting the DAG locally without additional coordination. 
Failing the checks does not introduce additional throughput or latency penalties; transactions simply finalize at their original commitment time.

To the best of our knowledge, \sys is the only asynchronous DAG-BFT protocol that provides early finality in such a manner.
We empirically compared \sys with Bullshark, the state-of-the-art asynchronous BFT protocol, in a geo-distributed setting across five AWS regions.
\sys attains up to 65\% lower consensus latency in the failure-free case, while achieving virtually equivalent throughput.
Even under adversarial conditions, \sys maintains its early finality benefits.
In the presence of a single node failure, \sys attains an approximate $50\%$ improvement;
even under maximum tolerable failures, \sys still maintains an approximate $24\%$ latency benefit.



\section{Preliminaries}
\label{sec: prelims}
We consider a set of $n$ nodes: $\Pi = \{p_1,\cdots,p_n\}$, where a dynamic computationally-bounded adversary can corrupt up to $f<\frac{n}{3}$ unique nodes.
Corrupted nodes are considered \textit{Byzantine}, whilst the remainder are regarded as \textit{honest}.
Byzantine nodes may exhibit arbitrary behavior, whereas honest nodes follow the protocol faithfully. We assume a public key infrastructure exists for node identity verification.

We adopt identical timing assumptions as Bullshark: we make no timing assumptions, and progress occurs even under full network asynchrony. 
Messages may be reordered or delayed arbitrarily, but are eventually delivered.

Since \sys utilizes identical dissemination and consensus mechanisms as Bullshark, we assume a reliable broadcast (RBC)\cite{Bracha1987AsynchronousBA}  primitive with agreement, integrity, and validity properties~\cite{bullshark}. 
We also assume a Global Perfect Coin (instantiated with threshold signatures)~\cite{Boneh2001ShortSF,Libert2014BornAR,Shoup2000PracticalTS} for randomized leader election to circumvent the FLP impossibility~\cite{flp}.



\section{Background and challenges}

The current state-of-the-art asynchronous BFT protocols are all DAG-based protocols, with Bullshark being the most performant implementation.
Since \sys utilizes the same state-of-the-art core consensus mechanism, we first examine its design before analyzing the fundamental challenges that prevent early finality.

\subsection{DAG Core}
\label{subsec:bullshark}

DAG-BFT protocols, such as Bullshark~\cite{bullshark}, decouple message dissemination from consensus. The dissemination protocol operates in sequential rounds $r = 1, 2, 3, \ldots$, where each node performs reliable broadcast (RBC)\cite{Bracha1987AsynchronousBA} of a message containing:

\begin{itemize}[itemsep=0pt, topsep=0pt, parsep=0pt,partopsep=0pt]
    \item A unique node identifier. 
    \item A set of client-submitted transactions.
    \item \textit{\textbf{Pointers}} to ($\geq 2f+1$) blocks from the previous round.
    \item Some additional metadata.
\end{itemize}

\textit{Transactions} represent atomic operations that read and modify key-value pairs in a shared state.

Upon successful completion of the RBC, the message becomes established as a block $b$.
\textbf{The RBC primitive provides two crucial guarantees}: eventual availability of a completed broadcast (block) to all correct nodes, and identical block delivery across all correct nodes (non-equivocation).
In effect, this strong primitive mitigates the most disruptive aspects of Byzantine behavior, ensuring that malicious nodes can only interfere with the protocol by remaining silent; much like a crashed node.

These blocks serve as vertices in a graph structure, with their pointers to previous-round blocks forming edges, collectively constructing the global DAG.
It is possible for a block to become \textit{orphaned} if no subsequent block in the following round points to it. We elaborate in \cref{app: orphaned} how \sys handles such blocks.

\subsubsection{Consensus Core}
\label{subsec: consensus in bullshark}

\begin{figure}
    \centering
    \includegraphics[width=.55\linewidth,trim=2.5cm .5cm 2.5cm .3cm]{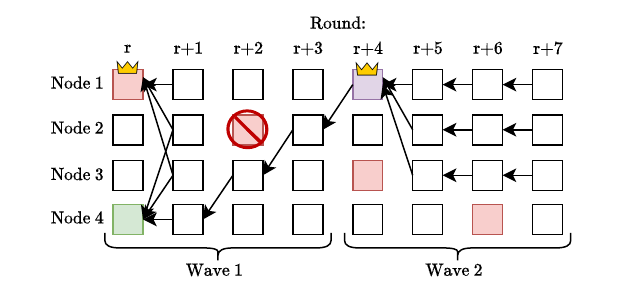}
    \caption{
    An example of latency disparities in Bullshark.  Only critical pointers are shown for clarity. Red blocks denote steady leaders, purple blocks denote fallback leaders, with crowns indicating committed leaders. The steady leader at round $r+1$ commits upon obtaining sufficient votes, while the green block at round $r$ must await inclusion in a future leader's causal history, incurring 8 additional rounds of delay. Were the green block itself a steady leader, it would commit at round $r+1$.
    \lijl{We can also briefly explain the DAG structure in this caption.}
    } 
    \vspace{-5pt}
    \label{fig: inequitable commit}
\end{figure}

Under network asynchrony, nodes may maintain divergent local views of the global DAG. However, since reliable broadcast precludes equivocation, Byzantine Atomic Broadcast is achieved by reaching consensus on a totally ordered list of blocks \textit{designated as leaders}.

The challenge lies in achieving agreement on which blocks are selected as leaders.
Bullshark specifically employs two leader types, where at most one type may commit within each 4-round \textit{wave}:

\begin{itemize}[itemsep=0pt,partopsep=0pt, topsep=0pt, parsep=0pt]
    \item \textbf{Steady leaders}: Selected deterministically via round-robin every 2 rounds. 
    \item \textbf{Fallback leaders}: Selected randomly using a global perfect coin every 4 rounds to circumvent the FLP impossibility~\cite{flp}.
\end{itemize}

When referring to a leader, we mean a block with that allocated designation.
\sys utilizes the same consensus mechanism as Bullshark.
\lijl{This is still introducing DAG background. Don't need to bring up Lemonshark.}\mich{I think a one-liner is fine, since there are many newer dag instantiations nowadays. specifying we are following a particular one seems better.}The leader selection process works as follows: in each wave, blocks produced by a particular node and its pointers constitute either a vote for a \textit{steady} or \textit{fallback} leader depending on whether progress was witnessed as part of its causal history in the previous wave.
A steady leader \textit{commits} when at least $2f+1$ blocks from the subsequent round have pointers to it. Fallback leaders are identified at the end of the wave, where blocks vote to commit by having paths to it. 
Optimal latency is enjoyed by a steady leader that obtains sufficient votes in the following round. It is important to note that at most one leader type may commit each wave. 

This leader commitment process is illustrated in \cref{fig: inequitable commit}, where the blocks are the result of reliable broadcast. The red blocks are steady leaders, and the purple blocks are fallback leaders. The steady leader in $r+1$ obtained sufficient votes and commits, while the missing steady leader in $r+3$ impedes the progress, leading to all block paths in the subsequent wave to be the fallback votes. The fallback leader in round $r+4$ is identified and voted upon at $r+7$.

\subsubsection{Ordering and Executing Transactions}
\label{subsec: Ordering and executing transactions}
\mich{Removed notations for greater clarity}
The consensus core (\cref{subsec: consensus in bullshark}) produces a totally ordered list of leaders. 
We define a block's \textit{causal history} as the set of blocks it has paths to, excluding those already committed by previous leaders. 
A node's view of any block's causal history may vary depending on its perception of the last committed leader. 
Each leader block's causal history is sorted using a deterministic algorithm and ordered in the same sequence as the initial totally ordered list. 

When a new leader $b'$ commits, its sorted causal history is appended to this sequence. Transactions within each block are then executed in this order, yielding \textit{finalized} outcomes for the corresponding transactions and blocks. 
Prior works~\cite{dag-rider,narwhal} employ different leader commitment mechanisms with varying election frequencies.
However, the resulting leader list and transaction execution semantics remain analogous.

\subsection{Inequitable Finality}
\label{subsec: inequitable finality}

While the ordering and execution mechanism described in \cref{subsec: Ordering and executing transactions} ensures consistency, it creates significant latency disparities between different types of blocks.
For a block to commit, it must either be part of a committed leader's causal history or be a leader itself. Since only a single leader may be elected every 2 rounds, most blocks experience longer commitment latencies than leaders.
\mich{fixed}
Even under optimal network conditions, non-leader blocks experience 1-2 additional rounds of delay compared to leaders, as they must await inclusion in a future leader's causal history. \lijl{We can add some explanation here, or add illustration in the figure}
\mich{I actually thought this was quite obvious from figure 1 haha}

This creates an inherent asymmetry where transaction latency depends not on network conditions or system load, but solely on the arbitrary assignment of transactions to leader versus non-leader blocks.

\subsection{Challenges in Achieving Early Finality }
\label{subsec:impossibility in bullshark}

\begin{figure}[t]
    \centering
    \includegraphics[width=0.6\linewidth,trim=1.5cm .5cm 1cm .3cm]{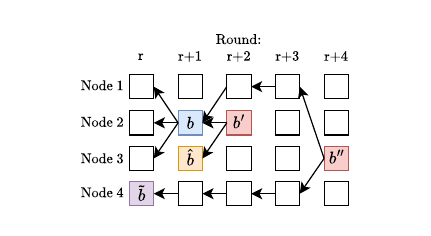}
    \caption{
    Both $\hat{b}$ and $\tilde{b}$ contain transactions conflicting with those in $b$. This figure demonstrates that Bullshark requires commitment to resolve ordering uncertainty for such blocks.
    }

    \vspace{-5pt}
    \label{fig: bullshark issues}
\end{figure}

\lijl{It is not obvious to a reader at this point that the desired goal is to have finality before commitment. Could introduce this goal first and what early finality means.}

This asymmetry raises a fundamental question: can we determine a transaction's finalized outcome before its containing block commits---circumventing the leader election bottleneck? This section demonstrates why early finality determination is unsafe in Bullshark, with formal proofs in \cref{lemma:  Impossibility of early finality in Bullshark}.

The core challenge stems from network asynchrony and Byzantine adversaries creating uncertainty about the complete set of conflicting transactions that may execute before a given block. 
We illustrate this using \cref{fig: bullshark issues}:  Consider Node 2 at round $r+1$ producing block $b$. Its local DAG view includes $b$'s causal history and at least $2f$ blocks from other nodes in round $r+1$. Due to asynchrony and potentially inactive Byzantine nodes, Node 2 proceeds to round $r+2$ without awaiting all $n$ blocks from round $r+1$.

This introduces ordering uncertainty: a block $\hat{b}$ outside Node 2's view may contain transactions affecting $b$'s outcome if $\hat{b}$ precedes $b$ in execution order. Since conventional DAG-BFT protocols impose no restrictions on transaction-to-block assignment, Byzantine adversaries can introduce such blocks in any round, modifying every system key.
If $\hat{b}$ obtains insufficient pointers ($<f+1$) in round $r+2$, we cannot guarantee that a future leader $b'$ (that has a path to $b$) at round $r'$ will have a path to $\hat{b}$. Thus, we must first witness $b'$. 

However, even this is insufficient. After observing $b'$ at round $r'$, uncertainty remains about whether $b'$ will commit; $b'$ may fail to secure sufficient votes, and a different eventual leader ($b''$) may have a different view regarding $\hat{b}$'s effect on $b$---potentially excluding $\hat{b}$ from its causal history entirely. 

These compounding uncertainties yield our main result: without fundamental modifications to how the DAG structure is perceived, early finality is unachievable under network asynchrony with Byzantine nodes.

\section{Problem Definition}

The inequitable finality described in \cref{subsec: inequitable finality} illustrates how arbitrary transaction-to-block assignments lead to the uneven latencies when obtaining finalized outcomes.
Lemonshark thus aims to \textit{determine a transaction's finalized outcome before it is committed}. Thereby allowing all blocks to yield lower latencies!
This section formalizes the notion of \textit{Early Finality} and establishes the theoretical framework for \sys.

\subsection{Refined Definition of Causal History}
We begin by providing a concrete definition of causal history, specifying the exact ordering mechanism:

\begin{definition}[(Sorted) Causal History of a Block]
\label{def: sorted causal history of a block}
For a block $b$, its causal history is the set of nodes that are part of a sub-DAG, where $b$ is the root, exclusive of the blocks that are in the previous known committed leader's causal history.
Its sorted causal history is obtained by applying Kahn's algorithm~\cite{Kahn1962TopologicalSO} to the sub-DAG and reversing the list, where ties (blocks of the same round) are broken deterministically.
We denote this list as $H_{b}=[\ldots,b]$.
\end{definition}

\begin{figure}[t]
    \centering
    \includegraphics[width=1.03\linewidth,trim=0.95cm .2cm 0cm .3cm]{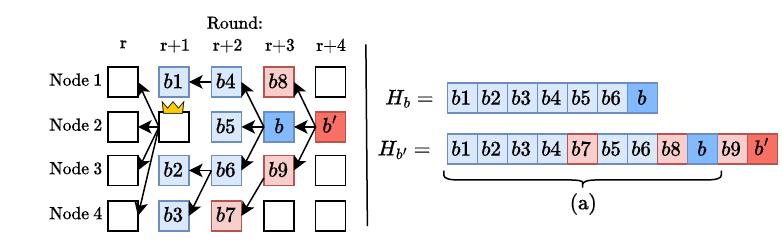}
    \caption{Blue blocks represent the causal history of $b$---excluding previously committed blocks. Blocks are ordered via breadth-first traversal by ascending round number, terminating at $b$. (a) Executing blocks in this order ($b_1,\dots,b$) yields the execution prefix $b'\langle b\rangle$.
    \lijl{Briefly describe the sorting rule in the caption.}\mich{done}} 
    \vspace{-5pt}
    \label{fig: reversekhan}
\end{figure}

\lijl{I wonder if this entire subsection can be merged into the background, and only discuss the new sorting rule.}
Unlike existing DAG-BFT protocols that accept \textit{any} deterministic ordering method, we impose a temporal constraint on the ordering: 
In a block's causal history, blocks from earlier rounds must be ordered before those from later rounds (\cref{fig: reversekhan} illustrates this concept). 
This ordering choice, while intuitive, is \textbf{critical} to \sys's early finality guarantees---arbitrary orderings would undermine our ability to evaluate transaction outcomes before commitment.

When a leader $b'$ commits, all blocks in $H_{b'}$ are committed and executed sequentially. Non-leader blocks are committed if and only if they appear in some committed leader's causal history.
For a committed leader $b'$, we say it \textit{commits} $b$ for all $b \in H_{b'}$ as well as all transactions in $b$. 
Given the strict monotonic ordering between committed leaders, for leaders $b'$ and $b''$ that are committed consecutively, all elements in $H_{b'}$ are committed before those in $H_{b''}$.
\mich{added this below to help with some clarity}

It is also important to note that a block $b$ can only be \textit{committed} by a single leader. Once committed by a leader, $b$ is excluded from the causal history of any other leader.


\subsection{Transaction and Block Outcomes}
\lijl{We can expand Figure 3 to also include illustrations for TO, BO, and potentially SBO.}
To evaluate early finality, we first define outcomes for transactions and blocks:
\begin{definition} [Transaction Outcome (\textit{TO})]
    For a block $b = [t_1,\ldots,t_i,\ldots,t_p]$, the transaction outcome (TO) of $t_i$ is the outcome of $t_i$ when executing transactions in the following order: $H_{b}[:-1]+[t_1,\ldots,t_i]$. We denote the prefix of $H_{b}$ exclusive of $b$ as $H_{b}[:-1]$. 
\end{definition}

\begin{definition}[Block Outcome \textit{(BO)}]
     For a block $b$ and its Sorted Causal History $H_b$, the Block Outcome (BO) of $b$ is the execution results of all transactions in $b$ after executing the blocks in the order of $H_b$.
\end{definition}

The rationale for evaluating transaction and block outcomes based on its own causal history is to eliminate intra-round dependencies; treating it as if it were a leader itself.

\subsection{Early Finality}
\lijl{We could add some intuitive notion of what early finality means first.}
To this end, we define \textit{Execution Prefix} as the outcome of intermediary blocks within a block's sorted causal history $H_{b'}$, computed prior to deriving the BO of $b'$. This serves as our reference point for comparison.

\begin{definition}[Execution Prefix (Block)]
    \lijl{Do we need to include transaction notations in this definition?}\mich{I believe this is the first place blocks and transactions are linked explicitly} For blocks $b',b: b=[t_1,\ldots,t_i,\ldots,t_p]$, where $b \in H_{b'}$. 
    The execution prefix $b'\langle b \rangle$ is the outcome of the transactions in $b$ when executing the prefix of $H_{b'}$ up to and including $b$ (\textit{i.e.}, $H_{b'}[0:index(b)]$ or $H_{b'}[0:index(b)-1]+[t_1,\ldots,t_i,\ldots,t_p]$ ). 
    We describe this as the execution prefix of $b$ with respect to $b'$.
   
\end{definition}

\begin{definition} [Execution Prefix (Transactions)]
    For blocks $b',b: b=[t_1,\ldots,t_i,\ldots,t_p]$, where $b \in H_{b'}$. 
    The execution prefix $b'\langle b(t_i) \rangle$ is outcome of transaction $t_i$ when executing the prefix of $H_{b'}$ right before $b$ and $[t_1,\ldots,t_i]$ (\textit{i.e.}, $H_{b'}[0:index(b)-1]+ [t_1,\ldots,t_i]$).
    We describe this as the execution prefix of $t_i$ with respect to $b'$.
\end{definition}

When the leader block $b'$ (where $b \in H_{b'}$) obtains sufficient votes, the execution prefix of each block/transaction in $H_{b'}$ with respect to $b'$ \textbf{becomes the \emph{finalized} immutable outcome. }
Suppose a non-leader block's BO matches the execution prefix with respect to the leader that commits it.
In that case, deriving the BO is functionally equivalent to deriving the block's eventual finalized outcome. We formalize this as Safe Transaction/Block Outcome:
\begin{definition}[Safe Transaction Outcome (STO)]
    For a particular transaction $t_i \in b$, we say that its transaction outcome is \textit{safe} if for a committed leader block $b': b\in H_{b'}$, the transaction outcome of $t_i$ is equivalent to the execution prefix $b' \langle b(t_i) \rangle $.
\end{definition}

\begin{definition}[Safe Block Outcome (SBO)]
    If all transactions within a block $b$ have STO, we say the block has a \textit{safe} block outcome (SBO). 
\end{definition}

Finally, we define our central objective of \sys:
\begin{definition}[Early Finality]
For a non-leader block $b$, early finality is achieved for $b$ if the SBO of $b$ may be determined before $b$ is determined to be committed. 
\end{definition}

\section{\sys}
\label{sec: lemon}





\begin{figure}[t]
    \centering
    \includegraphics[width=0.37\linewidth, trim=3cm .5cm 2.5cm 0cm]{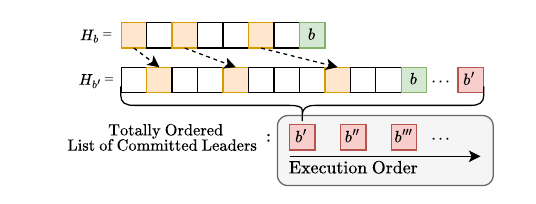}
    \caption{Orange blocks represent all uncommitted blocks from rounds prior to $b$'s round that might impact the \textit{finalized} outcome of $b$. The totally ordered list of committed leaders is generated by the DAG consensus core (\cref{subsec:bullshark}). \lijl{Missing a sentence to claim that SBO can be achieved in this situation}}
     \vspace{-5pt}
    \label{fig: lemon intuition}
\end{figure}

\sys builds upon the block dissemination and consensus mechanisms described in 
\cref{subsec:bullshark}, but restructures block content and reinterprets the DAG structure to enable early finality evaluation. 
As demonstrated in \cref{subsec:impossibility in bullshark}, early finality was unattainable because blocks potentially ordered before $b$---which may contain conflicting transactions---remain ambiguous until the leader committing $b$ is determined. 

\sys addresses this fundamental limitation through a key insight (illustrated with \cref{fig: lemon intuition}): Suppose a block $b$'s causal history includes all uncommitted blocks from prior rounds that can affect its execution outcome. When $b$ is committed by the leader $b'$, our ordering constraint (\cref{def: sorted causal history of a block}) ensures that only blocks from the same round as $b$ may affect $b$'s execution prefix with respect to $b'$. 
If those blocks indeed \textit{do not} affect $b$'s execution, SBO for $b$ can be evaluated before it is known to be committed. 

In \sys, this determination is made independently by each node using its local view of the DAG. 
It is important to note that \sys does not enforce the conditions for early finality;
rather, each node surveys its local DAG and identifies blocks that have met the criterion for SBO. 
All honest nodes will eventually make the same determination or, once the corresponding leader block $b'$ is committed, derive identical \textit{finalized} outcomes for all transactions in $b$. 
When SBO for $b$ is determined before commitment of $b'$, 
it may be delivered to the clients early, reducing perceived latency. 

In this section, we present the sufficient conditions that nodes can check locally to determine whether early finality can be achieved for a block, enabling safe execution results before commitment.

\lijl{I wonder if we could have an ``intuition'' subsection with an overview figure that explains this core concept. Something like Figure 1 in BullShark or Figure 3b in QuePaxa.}\mich{sorta done with the new fig}\lijl{TODO}\mich{section 5 and fig 4?}

\subsection{Sharded Key-Space}
\label{subsec: shared key-space}

As discussed in \cref{subsec:impossibility in bullshark}, each node must proceed after witnessing $\geq 2f+1$ blocks in round $r$, since waiting indefinitely for potentially inactive Byzantine nodes is impractical.

Consequently, when a node produces block $b$ in round $r+1$, it cannot assert that all blocks prior to round $r+1$ containing transactions affecting $b$'s execution prefix with respect to its eventual committing leader are included in $b$'s causal history.

\sys addresses this limitation by requiring each node to operate on a distinct, rotating partition of the key-space within each round. Specifically, we partition the key-space over which transactions operate into $n$ disjoint shards: $K = \{k_1,\ldots,k_n\}$.

In each round, only a single node may produce a block containing transactions that modify keys of a particular shard. 
Furthermore, the node-to-shard mapping rotates each round according to a publicly known schedule. For example, node $p_i$ in-charge of shard $k_i$ at round $r$ becomes in-charge of shard $k_{(i+1)\bmod n}$ at round $r+1$. 
A block containing transactions that exclusively write to keys in shard $k_i$ is designated as being \textbf{in-charge} of that shard. 

For notational clarity, we denote $b_i^r$ as a block from round $r$ that is in-charge of shard $k_i$. 
Clients may broadcast their transactions to all nodes, ensuring that a node in-charge of that key-space will be able to handle it immediately upon receiving it. 

This rotation scheme prevents censorship attacks and simplifies dependency tracking: block $b$ only needs to consider uncommitted blocks from previous rounds that operated on the same shard, rather than all potentially conflicting blocks. 
We assume the key-space partitioning scheme\cite{Okanami2020LoadBF, Han2022AnalysingAI} that achieves load balance while minimizing cross-shard transactions---The specific mechanisms for achieving optimal partitioning are beyond the scope of \sys. 

Key-space partitioning introduces cross-shard transaction complexity. To demonstrate \sys's generality, we define three transaction types:

\begin{itemize}[itemsep=0pt, topsep=0pt, parsep=0pt,partopsep=0pt]
    \item \textbf{Type $\alpha$: }Intra-shard transactions that read and write exclusively within shard $k_i$

    \item \textbf{Type $\beta$: } Cross-shard read transactions that read from multiple shards but write only to shard $k_i$ 

    \item \textbf{Type $\gamma$: } Atomic multi-shard transactions consisting of coordinated Type $\alpha/\beta$ sub-transactions that maintain serializability. 
    
\end{itemize}

These transaction types \textbf{cover essential database operations}\cite{Clark2024ValidatingDS,Tan2020CobraMT,Kingsbury2020ElleII}: local updates ($\alpha$), cross-shard reads ($\beta$), and atomic coordination ($\gamma$). Extensions to additional transaction types are possible, but beyond our current scope. 

For clarity, we focus on Type $\beta$ transactions that read from a single other shard, and Type $\gamma$ transactions that are sub-transaction pairs. Extensions to Type $\beta$ transactions reading from an arbitrary number of shards and Type $\gamma$ transactions as $n$-tuples are detailed in \cref{app: extending beta gamma}.

The remainder of this section examines how each transaction type may achieve early finality, along with various intricacies \sys must handle.


\subsection{Type $\alpha$: Intra-Shard Transactions}
\label{subsec: type a}
Type $\alpha$ transactions represent the optimal case for our \sys, performing both reads and writes exclusively within a single shard. This subsection demonstrates how blocks containing only Type $\alpha$ transactions can achieve SBO before commitment, building the foundation for the more complex transaction types.

We proceed inductively, first analyzing how the earliest uncommitted block in-charge of a shard achieves SBO. We then extend this analysis to subsequent blocks, carefully examining necessary conditions and potential edge cases that arise from evaluating early finality.
\lijl{The following subsections illustrate the protocol in an inductive proof style. I wonder if we can first present the entire requirement and a figure illustration without all the justifications, so the readers get a big picture first. Then dive into why this is safe.}\lijl{TODO}

\subsubsection{First Uncommitted Block}
\label{subsec: base case}


\begin{figure}[t]
    \centering
    \includegraphics[width=0.6\linewidth,trim=2cm .5cm 1.5cm 0.5cm]{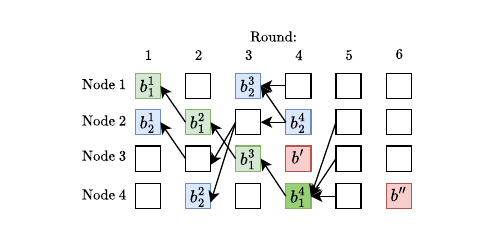}
    \caption{
    Consider a scenario where all blocks before round $1$ has committed. 
    The blue and green blocks are in-charge of two distinct shards ($k_1,k_2$), and the red blocks $b',b''$ are leader blocks.
    Suppose both $b_1^4,b_2^4$ persists in round $5$.
    For $b_2^4$, it has all uncommitted blue blocks as part of its causal history; however, each block does not point to the block of the same shard in the previous round. 
    The block outcome of $b_2^4$ will execute blue blocks in the order of: $(b^1_2,b^2_2,b^3_2,b_2^4)$.
    However, if $b^2_2 \in H_{b'}, b^1_2,b^3_2 \notin H_{b'}$ and $b'$ is committed, then the final order is $b^2_2,b^1_2,b^3_2,b_2^4$; therefore, $b_2^4$ does not have SBO. 
    In contrast, suppose $b_1^3$ has SBO, and $b_1^4$ points to it and persists in round $5$, its BO will execute the green blocks in the order: $(b^1_1,b^2_1,b^3_1,b_1^4)$. 
    Regardless of which subset of green blocks is in $H_{b'}$, the block outcome of $b_1^4$ will always be equivalent to the execution prefix $b'' \langle b_1\rangle$. }
    \vspace{-5pt}
    \label{fig: Type A-2}
\end{figure}

To inductively establish whether a block in-charge of a shard achieves SBO, we begin with the first uncommitted block in-charge of that shard, illustrated by block $b_1^1$ in \cref{fig:  Type A-2}.
\lijl{Is $b_x^y$ the best notation?}

This block possesses a unique advantage: if no leaders exist in the subsequent round and it is committed by a leader $b'$ from round 3 or later, it can guarantee its position as the first block in-charge of shard $k_1$ within the ordering of $H_{b'}$. Our ordering constraint (\cref{def: sorted causal history of a block}) ensures that no other block in-charge of shard $k_i$ may precede $b_1^1$ in $H_{b'}$.

However, $b_1^1$ must still ensure its eventual inclusion in a leader's causal history. This requirement is satisfied by achieving \textit{persistence}. A block \textit{persists} at round $r$ if $\geq f+1$ blocks from that round have paths to it. Since each block contains $\geq 2f+1$ pointers to blocks from the previous round, quorum intersection ensures that if a block persists at round $r$, all blocks from round $r+1$ onwards must have a path\footnote{This requirement aligns with the leader-election criterion for the partially synchronous version of 
Bullshark\cite{partial-bullshark}. } to it.

Consequently, if $b_1^1$ persists at round 2, any leader from round 3 onwards is guaranteed to have a path to it.  Therefore, absent a leader at round 2 and given that $b_1^1$ persists at round 2, its BO equals the execution prefix with respect to any leader from round $\geq 3$ that eventually commits it. 
Since the earliest such leader commits at round 4, exceeding the round at which $b_1^1$ persists (round 2), early finality is achievable for the first uncommitted block in-charge of a shard containing exclusively Type $\alpha$ transactions.





\subsubsection{Leader in the Next Round}
\label{para: leader exp}

\begin{figure}[t]
    \centering
    \includegraphics[width=0.55\linewidth,trim=2cm .3cm 2cm 0.5cm]{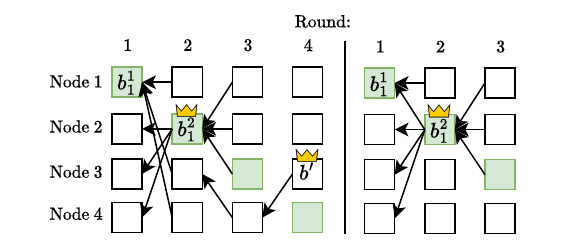}
    \caption{Suppose green blocks are in-charge of a unique shard ($k_1$). On the left portion of the figure, $b'$ eventually commits $b_1^1$. This figure illustrates two ways the restriction discussed in \cref{para: leader exp} may be bypassed.}
    \vspace{-5pt}
    \label{fig: lcheck}
\end{figure}
\lijl{Maybe have a separate subgraph besides the current Figure 4 to illustrate this?}\mich{Done}

In \cref{subsec: base case}, early finality was achieved for $b_1^1$ under the assumption that no leader exists in round 2. 
However, suppose $b_1^2$ is a leader that commits without having a path to $b_1^1$ (illustrated in the left portion of \cref{fig: lcheck}), where as $b_1^1$ is eventually committed by $b'\neq b_1^2$.
\lijl{Figure 6 shows a different scenario, so might be helpful to mention what $b'$ is in this case.}\mich{extended slightly}
In this case, we cannot assert that the BO of $b_1^1$ is safe, as execution follows the totally ordered list of committed leaders, possibly placing $b_1^2$ before $b_1^1$ in execution order.  
We now discuss sufficient conditions to circumvent this situation.

One approach is to wait until leader $b_1^2$ commits: 
Once $b_1^2$ is known to be committed, we can guarantee that no other block in-charge of $k_1$ will precede $b_1^1$ in $H_{b'}$. We formally show this in \cref{prop: leader check_2}. 
Since the earliest $b_1^2$ may commit is at round 3, and the next leader ($b'$) can only exist from round 4 onwards, early finality remains achievable for block $b^1_1$.

Alternatively, if $b_1^2$ has a pointer to $b_1^1$ (illustrated in the right portion of \cref{fig: lcheck}), then committing $b_1^2$ also commits $b_1^1$. Since $b_1^1$ is from an earlier round, our ordering method guarantees it executes before $b_1^2$.
This generalizes as follows: for a block $b_i^{r}$, if $b_i^{r+1}$ is a leader with a pointer to $b_i^{r}$, then $b_i^{r}$ executes before $b_i^{r+1}$ when committed. 

How can we determine which blocks may be elected as leaders in round $r+1$?  At each wave's beginning, the labeling of steady and fallback votes across rounds is predetermined, enabling identification of potential committed steady leaders per wave. However, determining potential fallback leaders is not possible \textit{a priori}.  Consequently, when a fallback leader could potentially commit, we conservatively assume any block in the wave's first round could become a committed fallback leader. In this case, we require $b_i^{r+1}$ to have a pointer to $b_i^{r}$, guaranteeing that $b_i^{r+1}$ will not execute before $b_i^{r}$. We formalize this requirement in \cref{prop: leader check_1}. 

We denote all conditions specified in this subsection as \textit{leader-checks} for shard $k_i$. Specifically, if either condition is satisfied for a block $b$ from round $r$ with respect to block $b_i^{r+1}$, we say that it passes the leader-checks for shard $k_i$. We summarize and illustrate the detailed conditions for leader-checks in \cref{alg:leader_check}.

\subsubsection{Subsequent Uncommitted Blocks}
\label{subsec: alpha extended}

We now demonstrate how blocks other than the earliest uncommitted ones in-charge of a shard may achieve early finality. Reusing the conditions specified in \cref{subsec: base case,para: leader exp} is insufficient. We illustrate this limitation by examining $b_2^3$ in \cref{fig: Type A-2}. 
Suppose $b_2^3$ persists in round 4, passes the leader-check for shard $k_2$ and is eventually committed by $b''$. Further suppose $b_2^2$ does not persist in rounds $3$, making it impossible to determine whether $b_2^2$ will be in $H_{b''}$ preceding $b_2^3$ in execution. 
This is because although $b_2^3$ ensures its eventual commitment by persisting in round 4, the fate of uncommitted blocks in-charge of $k_2$ prior to round 3 remains uncertain.

A straw-man solution would require all uncommitted blocks prior to round 3 and in-charge of $k_2$ to be part of $b_2^3$'s causal history, but we illustrate in~\cref{fig: Type A-2} that this may not be sufficient. Similar to the problem described in \cref{subsec:impossibility in bullshark}, having complete knowledge of all blocks locally does not necessarily imply the execution order. 
A leader's causal history might include only a subset of those blocks, resulting in a completely different execution order from the one in $b_2^3$'s BO. 

However, if a committed leader \textit{must} include a \textit{proper prefix} of those blocks, then we can ensure the eventual execution is consistent with $b_2^3$'s BO.
\sys addresses this challenge by requiring block $b_i^r$ to have a pointer to $b_i^{r-1}$, where the latter has SBO. This creates a recursive path of blocks from $b_i^r$ to the oldest uncommitted block in-charge of $k_i$ in round $\hat{r}$, traversing blocks from rounds $\hat{r}+1$ to $r-1$ that are in-charge of $k_i$.
When such a path exists, we demonstrate in \cref{fig: Type A-2} that this structure is sufficient to ensure that any committed leader must include a proper prefix of this chain; it may not contain an arbitrary subset of blocks from this chain within its causal history.

More generally, if $b_i^r$ points to $b_i^{r-1}$ where the latter has SBO, persists in $r+1$, and passes the leader-check for $k_i$, we can guarantee that no block in-charge of $k_i$ that is not contained in $b_i^r$'s causal history can execute before $b_i^r$. Therefore, we can ensure that the BO of $b_i^r$ will match the execution prefix with respect to the leader that eventually commits it and thus have SBO.

\begin{algorithm}[t]
\footnotesize
\caption{$\alpha$-STO Eligibility Check}
\label{alg:alpha-sto-check}
\begin{algorithmic}[1]
\Function{$\alpha$-StoCheck}{$t \in b_i^r$}
    \If{$\exists t' \in DL_r$ where $t$ modifies same key as $t'$}
        \State \Return \textbf{false}
    \EndIf
    \If{\textbf{not} \Call{LeaderCheck}{$b_i^r, k_i$}}
        \State \Return \textbf{false}
    \EndIf
    \State \Return $\bigl($($b_i^r$ is oldest uncommitted block in-charge of $k_i$) 
    \textbf{or} ($b_i^r$ points to $b_i^{r-1}$ and $b_i^{r-1}$ has SBO)$\bigl)$
    \textbf{and} ($b_i^r$ persists in $r+1$)
\EndFunction
\end{algorithmic}
\vspace{0.5em}
\hrule
\vspace{0.3em}
\footnotesize
\textbf{Note:} $DL_r$ will be elaborated further when we discuss Type $\gamma$ transactions in \cref{sec: type c diff}. LEADERCHECK() is described with \cref{para: leader exp} and illustrated in \cref{alg:leader_check}.
\end{algorithm}

Since the earliest a non-leader block $b_i^r$ can be committed is in round $r+2$, while learning that it persists can be achieved in $r+1$, early finality remains possible. We illustrate all sufficient conditions for a Type $\alpha$ transaction to have STO in \cref{alg:alpha-sto-check} and provide formal proofs in \cref{lemma: STO for alpha}.
We also demonstrate in \cref{prop: persist bound} that since a substantial subset of at least $\frac{3f+2}{2}$ blocks from each round \textit{must} persist in the subsequent round, early finality for some blocks remains achievable even in the presence of byzantine adversaries and network asynchrony.






                                              
\subsection{Type $\beta$: Cross-shard Read Transactions }
\label{subsec: type beta}

Type $\beta$ transactions differ from Type $\alpha$ transactions by reading from a different shard that they write to. 
Since Type $\beta$ transactions continue to write to the shard that their containing block is in-charge of, they retain the same requirements as Type $\alpha$ transactions. Additionally, we must ensure that the perceived read value when evaluating its TO matches that obtained when evaluating the execution prefix of the transaction with respect to the leader that eventually commits it.

In this subsection, we specifically examine how a particular block from round $r$ may exhibit characteristics that accommodate possible modifications to the read value occurring before, during, and after $r$. We then identify conditions sufficient for a Type $\beta$ transaction within it to achieve early finality \lijl{not a complete sentence?}.\mich{fixed}
To aid in our illustration, consider a Type $\beta$ transaction $t$ in a non-leader block $b_i^r$. It reads from $k_j: k_j\neq k_i$ and is eventually committed by a leader block $b'$\lijl{potentially helpful to mention that $k_j$ is a different shard from the writing shard $k_i$}.\mich{fixed}

\subsubsection{Read Value Before $r$}
\label{para: before r}

Similar to \cref{subsec: alpha extended}, we must ensure that all possible uncommitted blocks modifying $k_j$ prior to round $r$ execute before $b_i^r$. This is achieved by having $b_i^r$ point to $b_j^{r-1}$, where the latter has SBO. We can then assert: 
no uncommitted blocks in-charge of $k_j$ existing before round $r$ will execute before $b_i^r$ \lijl{having some question about the expression here. can discuss}. 
Alternatively, if the oldest uncommitted block in-charge of $k_j$ is at round $\geq r$, then no block in-charge of $k_j$ prior to round $r$ may be ordered before $b_i^r$ in execution.

\subsubsection{Read Value During $r$}
\label{para: during r}

\begin{figure}[t]
    \centering
    \includegraphics[width=0.5\linewidth,trim=2cm .3cm 1.5cm .3cm]{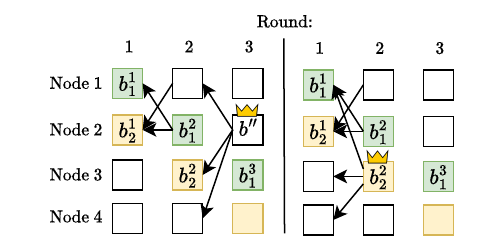}
    \caption{Suppose the green and yellow blocks are in-charge of two distinct shards, $k_1,k_2$ respectively, no blocks exist before round $1$, and leader $b'$ eventually commits $b^2_1$. 
    This figure illustrates the two cases in which $b_2^2$ may be committed while $b_1^2$ is not---a case for Type $\beta$ transactions.} 
    \vspace{-5pt}
    \label{fig: Type B}
\end{figure}

Our ordering procedure (\cref{def: sorted causal history of a block}) orders blocks of the same round in a mere deterministic order. Therefore, it is possible that $b_j^r$ exists in $H_{b'}$ and is ordered before $b_i^r$. 
We do not make any ordering assumptions for blocks within the same round. 
Yet, we have to ensure that $b_j^r$ does not affect the execution prefix of $t$ with respect to $b'$.\lijl{the notation here is a bit heavy}

Instead, if a node were to witness $b_j^r$ in its local DAG, it has to verify if it contains any transactions that write to the key that $t$ reads.
If it does not, then it need not worry if any transaction in $b_j^r$ were to affect the execution prefix of $t$ with respect to $b'$.
However, if it does, $b_j^r$ must be committed by a different leader $b''$ in a round earlier than that of $b'$, so as not to affect the execution prefix of $t$ with respect to $b'$. This condition is analogous to the first one described in \cref{para:  leader exp}. 
There exist some intricacies in \textbf{when} exactly $b_j^r$ is committed. We illustrate 2 cases in  \cref{fig: Type B}, focusing on the green block $b_1^2$.
Suppose $b_1^2$ has already satisfied all requirements described in \cref{para: before r} and contains a Type $\beta$ transaction $(t)$ that reads from the yellow shard $k_2$.

\mich{slightly rephrased this section due to reviewer 3's comments}
\lijl{Not very clear how the following two paragraphs translate to our protocol}
In the first case (depicted in the left portion of \cref{fig: Type B}), $b_2^2\in H_{b''}$, where $b''$ is a committed leader.
Furthermore, $b_2^2$ does not have a path to $b^1_2$ and not $b_1^2$ \lijl{they are in the same round, how can they have a path?}\mich{i made a mistake, its $b^1_2$}, but $b_1^2$ has a path to $b_2^1$ (where the latter has SBO).
Per \cref{para: before r}, $b''$ must therefore have a path to $b_2^1$. At this point, all blocks before round 3 that are in-charge of $k_2$ have been committed, and we can be assured that no block in-charge of $k_2$ before round 3 will supersede $b_1^2$ in $H_{b'}$.
In the second case (depicted in the right portion of \cref{fig: Type B}), $b_2^2$ \textit{is} a committed leader. Regardless of whether $b_2^2$ has a path to $b_2^1$, once $b_2^2$ is known to be committed, $b_2^1$ satisfies the second condition of leader-checks on $k_2$ (per \cref{para: leader exp}) and $t$ achieves STO. Similarly to the previous case, we can be assured that no block in-charge of $k_2$ before round 3 will supersede $b_1^2$ in $H_{b'}$ that does not exist in $H_{b_1^2}$.

\subsubsection{Read Value After $r$}
\label{para: after r}

\begin{algorithm}[t]
\footnotesize
\caption{$\beta$-STO Eligibility Check}
\label{alg:beta-sto-check}
\begin{algorithmic}[1]
\Function{$\beta$-StoCheck}{$t \in b_i^r, k_j: k_j \neq k_i$}
    \If{$\exists t' \in DL_r$ where $t$ modifies same key as $t'$ \textbf{or} \textbf{not} \Call{$\alpha$-StoCheck}{$t \in b_i^r$}}
        \State \Return \textbf{false}
    \EndIf
    \If{\textbf{not} $\bigl($($b_j^r$ is oldest uncommitted block in-charge of $k_i$) \textbf{or} ($b
    _i^r$ points to $b_j^{r-1}$ and $b_j^{r-1}$ has SBO)$\bigl)$}
        \State \Return \textbf{false}
    \EndIf
    \If{$b_j^r$ modifies key that $t$ reads \textbf{and} $b_j^r$ not yet committed}
        \State \Return \textbf{false}
    \EndIf
    \State \Return \Call{LeaderCheck}{$b_i^r, k_j$} \textbf{or} \textbf{not} $\exists t' \in b_j^{r+1}$ that modifies what $t$ reads
\EndFunction
\end{algorithmic}
\vspace{0.5em}
\hrule
\vspace{0.3em}
\footnotesize
\textbf{Note:} $DL_r$ will be elaborated further when we discuss Type $\gamma$ transactions in \cref{sec: type c diff}. LEADERCHECK() is described in \cref{para: leader exp} and illustrated with \cref{alg:leader_check}.
\end{algorithm}

Lastly, similar to \cref{para: leader exp}, we must consider the case where block $b_j^{r+1}$ may contain transactions that modify the key $t$ reads and is committed before $b'/b_i^r$. In such a case, we simply perform the same leader-checks for $b_i^r$ but on shard $k_j$. 
These checks can be omitted if $b_j^{r+1}$ does not contain transactions that modify the key $t$ is reading, since even if it precedes $b_i^r$ in execution, it does not affect the execution prefix of $t$ with respect to $b'$.

Collectively, with the conditions mentioned in \cref{para: before r,para: during r}, we can ensure that blocks in-charge of $k_j$ from rounds before, during, and after $r$ will not impact the execution prefix of $t$ with respect to $b'$. We illustrate the sufficient conditions for a Type $\beta$ transaction to have STO in \cref{alg:beta-sto-check}. Furthermore, all these conditions may be evaluated by analyzing the local DAG before $b'$ is committed; therefore, early finality for Type $\beta$ transactions is achievable. We provide complete formal proofs in \cref{lemma: STO for beta}.


\subsection{Type $\gamma$: Coordinated Cross-Shard Transactions}
\label{subsec: type gamma}

Coordinated operations across multiple shards are a fundamental requirement in distributed databases. Atomic and serializable tuples of transactions enable consistent writes across shards. 
Combined with Type $\alpha,\beta$ transactions, these three transaction types are sufficient to express the core operations required by most distributed database workloads.

A Type $\gamma$ transaction is a specialized tuple of Type $\alpha/\beta$ sub-transactions that are \textit{atomic} and \textit{serializable as a tuple}. As mentioned in \cref{subsec: shared key-space}, we will discuss Type $\gamma$ transactions as a pair of sub-transactions. 

Atomicity is guaranteed if all sub-transactions are eventually executed; this can be achieved by having both sub-transactions include each other as part of its metadata. Therefore, if one is part of a block, the other will be known by all nodes and eventually be incorporated into the DAG and committed as well.

The challenge lies in ensuring that Type $\gamma$ sub-transactions are \textit{serializable as a pair}. We illustrate this with an example: Consider a Type $\gamma$ transaction that swaps the values of two keys $k_j^1 \mapsto \text{\textit{apple}}$ and $k_i^1 \mapsto \text{\textit{orange}}$, where $k_j^1 \in k_j, k_i^1 \in k_i$ and $k_i \neq k_j$.
In this particular instance, the Type $\gamma$ transaction comprises two Type $\beta$ sub-transactions:

\begin{enumerate}[itemsep=0pt, topsep=0pt, parsep=.3pt]
    \item Sub-transaction 1: Read $k_j^1$ and write it into $k_i^1$.
    
    \item Sub-transaction 2: Read $k_i^1$ and write it into $k_j^1$.
\end{enumerate}

The desired outcome should be $k_j^1 \mapsto \text{\textit{orange}}, k_i^1 \mapsto \text{\textit{apple}}$. 
However, as the sub-transactions execute independently and sequentially, the result will be both keys having identical values. This is the case even if both transactions are adjacent in the sorted causal history of a leader that commits them both. 
Furthermore, a third transaction (such as writing $\text{\textit{mango}}$ to $k_i^1$) may be ordered between the two sub-transactions, causing the final outcome to deviate from the desired result. 

Therefore, the sub-transactions that constitute a Type $\gamma$ transaction must be executed atomically at the \textit{same time} for the desired outcome to be achievable.

\subsubsection{Ordering and Executing Type $\gamma$ Sub-Transactions}
\label{subsec: ordering and executing gamma}

Enforcing that both sub-transactions of a Type $\gamma$ transaction are executed together is non-trivial, especially since both components of a Type $\gamma$ transaction may exist in different blocks, potentially in different rounds and within different committed leaders' causal histories.

Type $\alpha$ and $\beta$ transactions follow the commitment and execution procedure described in \cref{subsec: Ordering and executing transactions}, while Type $\gamma$ transaction \textit{executions deviate slightly}.
Consider two sub-transactions that constitute a Type $\gamma$ transaction, $t_1 \in b_i^{r_1}$ and $t_2 \in b_j^{r_2}$, where $b_i^{r_1} \neq b_j^{r_2}$. We consider 3 cases:

\begin{figure}[t]
    \centering
    \includegraphics[width=0.6\linewidth,trim=1.9cm .5cm 1.5cm .3cm]{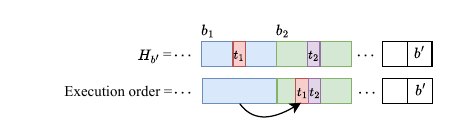}
    \caption{This figure shows how a Type $\gamma$ transaction ($t_1,t_2$) affects the execution order.}
    \vspace{-5pt}
    \label{fig: type c1}
\end{figure}

\begin{itemize}[itemsep=0pt, topsep=0pt, parsep=0pt,partopsep=0pt]
    \item Consider $r_1=r_2$ and $b_i^{r_1}, b_j^{r_2}$ are committed by the same leader. Suppose $b_i^{r_1}$ is ordered before $b_j^{r_2}$ (per \cref{def: sorted causal history of a block}). $t_1$ will not be executed with the other transactions in $b_i^{r_1}$. Instead, $t_1$ is executed concurrently with $t_2$ in $b_j^{r_2}$. We illustrate this example in \cref{fig: type c1}.

    \item Consider $r_1 < r_2$ and $b_i^{r_1}, b_j^{r_2}$ are committed by the same leader. Since $b_i^{r_1}$ is from an earlier round, it is ordered before $b_j^{r_2}$ in the leader's sorted causal history. Therefore, similar to the previous case, $t_1$ is executed together with $t_2$ in $b_j^{r_2}$.

    \item Consider $b_i^{r_1}, b_j^{r_2}$ are committed by different leaders (with $b_i^{r_1}$ committed earlier). Transaction $t_1$ will not execute until $t_2$ is committed, and $t_1$ will similarly be executed together with $t_2$ in $b_j^{r_2}$.
\end{itemize}

Essentially, although sub-transactions may exist in distinct blocks, they are reordered such that they execute together, as-if both exist in \textit{one} of the 2 blocks. 
For illustration, we denote that block as the \textit{\textbf{prime}} block and the corresponding transaction (that exists physically in the block) as the prime sub-transaction. The other block/transaction is likewise referred to as the non-prime block/sub-transaction.

A challenge of achieving early finality for Type $\gamma$ transactions is determining which transaction is \textit{prime}. This is mainly because it requires us to witness both blocks and identify which leader's causal history they are part of.

\sys's technique to achieve early finality for Type $\gamma$ transactions is straightforward: if we can guarantee that both sub-transactions are part of the same leader's causal history, we can determine which sub-transaction is prime. Subsequently, from the position of the prime sub-transaction, we evaluate if STO for both sub-transactions can be achieved before it is committed.
\lijl{Paper prior to this point determines early finality at the block level. Type $\gamma$, however, does it at the transaction level. Might be worth noting this, and briefly explain how this plays with a block that contains other types of transactions.}
\mich{i think beta does so on a transaction level too, only alpha does it on a block level, since its the simplest.}

\subsubsection{Same Round, Same Leader}
\label{para: same leader causal history} 

We now describe how early finality for Type $\gamma$ transactions may be achieved per the first case described in \cref{subsec: ordering and executing gamma}.
Consider two sub-transactions that constitute a Type $\gamma$ transaction, $t_1$ and $t_2$, existing in blocks $b_i^{r_1}$ and $b_j^{r_2}$ respectively.

\paragraph{Same Leader's Causal History}
The requirement that both blocks exist within the same leader's causal history is readily satisfied.
This condition is trivially met when a leader in round $\max(r_1,r_2)+1$ contains both blocks within its causal history. Conversely, when no leader from a round prior to $\max(r_1,r_2)+1$ has either block in its causal history, the persistence of both blocks in round $\max(r_1,r_2)+1$ ensures their eventual inclusion in the same committed leader's ($b'$) causal history. We establish this property in \cref{prop: 2 blocks 1 leader}.

Suppose $b'$ is the leader that ultimately commits both blocks. Under the assumption $r_1 = r_2$ and $b_i^{r_1}$ precedes $b_j^{r_2}$ in the ordering of $H_{b'}$: $t_2$ constitutes the prime sub-transaction. 

\paragraph{Condition for STO}
Firstly, we require that $t_1$ and $t_2$ be evaluated to have STO independently---treating each as an autonomous Type $\alpha$ or $\beta$ transaction. 
When the non-prime sub-transaction $t_1$ achieves STO independently, this indicates that all uncommitted blocks capable of influencing its outcome that exist in rounds $<r_1$ are contained within $b_i^{r_1}$. 

Similarly, satisfying the leader-checks ensures that blocks in the subsequent round $r_1+1=r_2+1$ will not supersede them during execution according to our sorting procedure established in \cref{def: sorted causal history of a block}.
Given that both $b_i^{r_1}$ and $b_j^{r_2}$ are elements of $H_{b'}$, this guarantees that $t_1$ would maintain STO if it were the sole transaction in $b_j^{r_2}$. 

When no additional transactions exist in $b_i^{r_1}$ and $b_j^{r_2}$, this condition alone suffices for STO evaluation of $t_1$ and $t_2$. However, given that this scenario is not frequently realized, we must account for the potential influence of other transactions in $b_i^{r_1}$ and those preceding $t_2 \in b_j^{r_2}$ on the outcome of $t_1$. 
To address this complexity, we impose a comprehensive constraint: all remaining transactions in both $b_i^{r_1}$ and $b_j^{r_2}$ must have STO. This requirement enables the deterministic assessment of their impact on $t_1$ in advance, enabling us to produce STOs for both $t_1$ and $t_2$ that match the execution prefix of $t_1$ and $t_2$ with respect to $b'$.

Given that every transaction within $b_i^{r_1}$ and $b_j^{r_2}$ can be determined to possess STO in round $\max(r_1,r_2)+1$ before commitment, early finality remains achievable for Type $\gamma$ transactions. The complete formal proof is presented in \cref{lemma: gamma 1}.

We relegate the more intricate case where $r_1 \neq r_2$ to the appendix (\cref{lemma: gamma 2}). Nevertheless, the underlying principle is intuitive: if we can reduce the problem to the one described in this section, then early finality is achievable as well.


\subsubsection{Different Leaders, Delay List}
\label{sec: type c diff}

As highlighted by the third case specified in \cref{subsec: ordering and executing gamma}---it is entirely possible that $b_i^{r_1}$ and $b_j^{r_2}$ are committed by two distinct leaders.
In this particular case, we are unable to evaluate the TO of the non-prime sub-transaction (suppose it is $t_1 \in b_i^{r_1}$) until the prime counterpart is committed. This inability to evaluate the TO of $t_1$ makes it impossible to derive the outcome of subsequent transactions from shard $k_i$ that read or write to the same key that $t_1$ modifies. This naturally implies that STO would be impossible to derive as well.

We handle this constraint with a blacklisting approach: by adding $t_1$ in the Delay List for round $r_1$ ($DL_{r_1}$).
The delay list $DL_r$ contains transactions from rounds up to $r$ that are appended in this manner. If a transaction in a block from round $r$ reads or writes to the same key modified by a transaction that exists in $DL_r$, that transaction automatically becomes ineligible to achieve STO. 
To ensure that the delay list is comprehensive, we show in \cref{prop: DL} that when a node deems that a transaction were to fulfill all other requirements for STO, its view of the Delay List must include all transactions that may possibly affect it.

A Type $\gamma$ sub-transaction is only removed from $DL$ if both sub-transactions are committed, or when the prime sub-transaction is evaluated to have STO. We elaborate on how the latter is possible in \cref{lemma: gamma 2}, where $t_1$ and $t_2$ are part of the same leader's causal history but in different rounds.

An unfortunate consequence of the Delay List mechanism is that if $t_1$ is committed earlier, we are delaying its execution until $t_2$ is committed, possibly incurring rounds of additional delay. Nevertheless, because Type $\gamma$ transactions are assumed to be relatively rare (due to how the key-space is partitioned into shards), we expect this issue to have only a negligible impact in practice.

\section{Pipelined Dependent Transactions}

Orthogonal to our main insight in \cref{sec: lemon}, we observe that clients with dependent transactions (where subsequent transactions rely on the outcomes of prior transactions) typically face multiplicative latency penalties. By providing tentative speculated results earlier, subsequent transactions may be introduced  into the DAG sooner (\cref{fig: dependent transactions}). When speculation succeeds, this approach achieves lower latencies; when incorrect, subsequent transactions are merely aborted and latency reverts to the baseline case.
We provide a more detailed analysis, experimental results, and the integration with \sys's core contributions in \cref{app: pipelined}.

\begin{figure}[t]
    \centering
    \includegraphics[width=0.50\linewidth,trim=2cm 0.7cm 2cm .5cm]{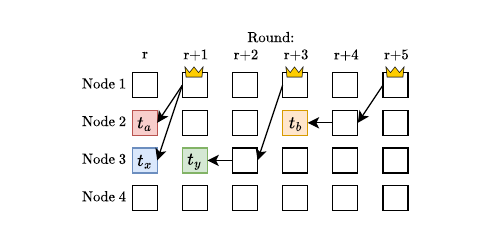}
    \caption{How speculation can aid in reducing latencies for dependent transactions. Suppose the pairs $t_a,t_b$ and $t_x,t_y$ are dependent. However, we allow $t_y$ to be submitted earlier along with a speculated result for $t_x$. $t_y$ execution is contingent on whether that a speculated result is equivalent to the finalized outcome of $t_x$}
    \vspace{-6pt}
    \label{fig: dependent transactions}
\end{figure}

\section{Implementation}

We implement Asynchronous Bullshark by forking Bullshark's \cite{bullsharkImplementation} open source repository, which itself 
forks from the Narwhal project\cite{narwhalImplementation} and utilizes Narwhal's DAG structure for its consensus core. 
We then implement \sys on top of our Asynchronous Bullshark implementation to interpret the DAG differently for early finality. 
The implementation is written in Rust and uses \texttt{tokio}~\cite{tokio} for asynchronous networking. \texttt{ed25519-dalek}~\cite{Dalekcrypto} is used as the cryptographic library and RocksDB~\cite{RocksDB} for persistent storage of the DAG.

\begin{figure*}[t]
    \centering
    \includegraphics[width=.65\linewidth,trim=1cm .5cm 1cm 0cm]{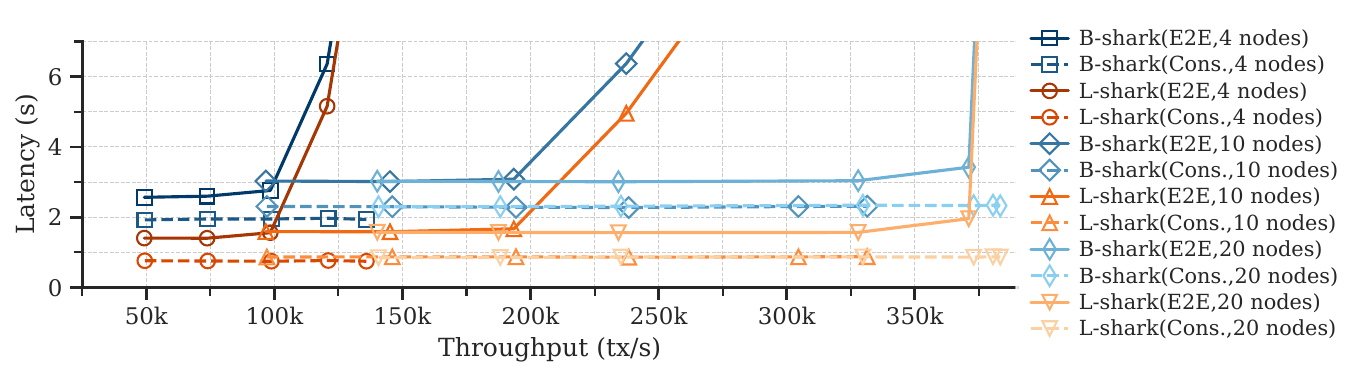}
    \caption{Performance of \sys with Type $\alpha$ transactions vs Bullshark with no faults, varying the number of nodes.}
    \vspace{-15pt}
    \label{fig:exp:max throughput}
\end{figure*}

\paragraph{Baselines:}

We compare \sys explicitly against Bullshark~\cite{bullshark} as it is the only state-of-the-art asynchronous DAG-BFT with a complete open-source implementation. We exclude Tusk~\cite{narwhal} from comparison as Bullshark's greater leader election frequency yields 33\% lower latencies. We also exclude Agreement-on-Common-Set asynchronous BFT protocols like Dumbo~\cite{Lu2020DumboMVBAOM}, as Tusk demonstrates substantial ($20\times$) latency improvements over these approaches.
Uncertified DAG BFT protocols either lack practical implementations~\cite{cordial} or operate under different network synchrony assumptions~\cite{Babel2023MysticetiRT}, which makes direct comparison with \sys inappropriate.

\section{Evaluation}

In theory, early finality in \sys should improve transaction latency while maintaining throughput and scalability of a DAG-BFT protocol.
We now evaluate \sys to empirically demonstrate the benefit of early finality.
In particular, we aim to answer the following questions:

\begin{itemize}[itemsep=0pt, topsep=0pt, parsep=0pt,partopsep=0pt]
    \item Can \sys consistently achieve lower latency through early finality, while maintaining high throughput and scalability?
    \item Can \sys sustain its performance even for cross-shard transactions?
    \item Can \sys maintain its latency benefit even in the presence of failures? 
\end{itemize}

We use two types of latencies in the evaluation.
\textbf{Consensus latency} refers to the time taken for a block to be finalized after its reliable broadcast.
\textbf{End-to-end (E2E) latency} refers to the duration for a transaction within a block to be finalized after it has been generated by the client.
For Type $\gamma$ transactions, both E2E/consensus latencies are recorded with respect to the latest submitted sub-transaction onto the DAG. 

\paragraph{Experimental Setup:}
We deploy our testbed on AWS, utilizing \texttt{m5.8xlarge}~\cite{AmazonEc2} instances distributed across 5 regions: N. Virginia (us-east-1), N. California (us-west-1), Sydney (ap-southeast-2), Stockholm (eu-north-1), and Tokyo (ap-northeast-1). 
Each machine has 10Gbps bandwidth, 32 vCPUs running at 3.1GHz, and 128GiB memory. These machines are selected as they mirror those deployed by production blockchains and represent commodity servers.

\mich{todo: explain the geolocated latencies part:}
As described in \cref{subsec: shared key-space}, shard-rotation is deterministic and public, but clients should broadcast transactions to all nodes. 
Following previous works~\cite{narwhal,bullshark,shoal++}, our clients connect locally to each instance and continuously send streams of 512B ``\texttt{nop}'' transactions.  This approach isolates consensus performance from two orthogonal factors: (1) client geolocation, which would otherwise introduce variable network delays based on client proximity to the instances\footnote{We measured a maximum of $\sim$300ms network latency between the most distant instance pair in our deployment. Even accounting for this maximal penalty, \sys achieves superior end-to-end latencies in all experiments.}, and (2) transaction execution overhead, which would confound measurements of consensus latency. 

To evaluate cross-shard operations, we mark each block's meta at dissemination to denote transaction types it carries (e.g., whether it contains a set of type $\beta$ transactions that read from other shards).

Narwhal proposes scaling dissemination by utilizing an additional ``worker-layer'' that performs the first step of reliable broadcast in batches. Blocks in the DAG then contain hashes representing batches. In our experiments, we utilize batch sizes of 500KB and block sizes of 1000B\footnote{Since each hash digest is 32B long, a 1000B block represents approximately 32K transactions}. 
We utilize a \textit{leader-timeout} value of 5 seconds; if a steady leader should exist in a round, each node waits that duration before proceeding without the leader's block in its local DAG.

As discussed in \cref{subsec: consensus in bullshark}, Byzantine nodes are unable to equivocate due to the reliable broadcast primitive. Therefore, in our experiments, we simulate only \textit{crash-faults}. 
We deviate from past experimental methodologies~\cite{bullshark,narwhal}, normalizing failure behavior by randomly selecting faulty nodes and randomizing steady leader selection. We elaborate in \cref{app: experiment discussion} on why we believe this approach is fairer.
We also elaborate in \cref{app: orphaned} how missing blocks caused by crash faults can be identified by honest nodes. 
All measurements represent the average of 3 independent runs of 3 minutes.

\subsection{Single Shard Transactions}

\sys's latency is optimal when all transactions are Type $\alpha$. As such, we compare that performance with Bullshark with no crash-faults present. 
Since every non-leader block in \sys should qualify for early finality after witnessing sufficient blocks in the subsequent round, consensus latency for all blocks approaches optimal at approximately $65\%$ lower than Bullshark (\cref{fig:exp:max throughput}). 
This performance is sustained even when we increase client transaction rates and committee sizes. Eventually, client transaction rates exceed consensus capacity, causing queues to form indefinitely and latencies to spike.  

For the remainder of our experiments, we utilize a baseline transaction rate of 100K-tx/s and a committee size of 10.
This represents a moderate load that effectively stresses the consensus mechanism while remaining within capacity.


\subsection{Cross-Shard Transactions}
\label{sec: exp type b}

\begin{figure}[t] 
    \centering
    \includegraphics[width=\linewidth]{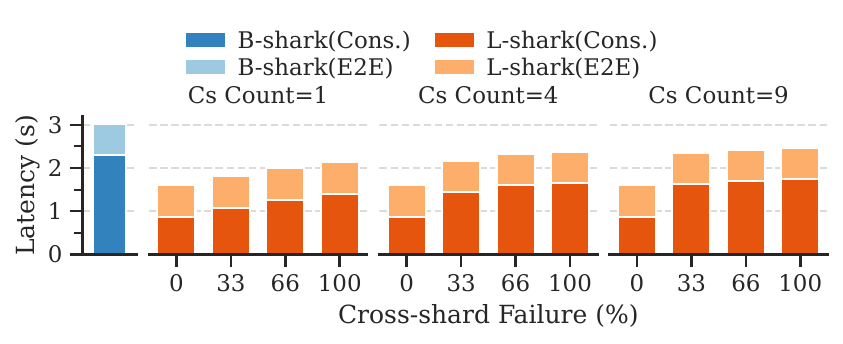}
    \vspace{-20pt}
    \caption{Performance of \sys with Type $\beta$ transactions, 
    while varying the amount of cross-shard activity and STO failure rates; ``Cs Count'' refers to ``Cross-shard Count''.}
    \vspace{-5pt}
    \label{fig:exp2}
\end{figure}

Since Type $\gamma$ transactions consist of pairs of Type $\alpha/\beta$ transactions, their performance closely correlates with that of their constituent sub-transactions in failure-free scenarios.
Therefore, we evaluate cross-shard behavior by simulating Type $\beta$ transactions and $\gamma$ sub-transactions.
We configure 50\%\footnote{We vary this proportion in \cref{app: add. results}} of all blocks to contain cross-shard transactions that either read from a number of shards equal to ``Cross-shard Count'' or have sub-transactions distributed across that many shards. 
When a block in round $r$ contains cross-shard transactions, we randomly determine the number of shards it interacts with (uniformly from 0 to ``Cross-shard Count'') for either reads or sub-transaction placement.
For each shard from which it reads, ``Cross-shard Failure'' denotes the probability that the read key is modified by another block in round $r$, or that the companion sub-transactions do not exist in the same round.

Given the relatively synchronous nature of the AWS network, the primary factor preventing a Type $\beta$ transaction $t$ (from a block in round $r$ reading from $k_j$) from obtaining STO is not pointing to blocks from the previous round (\cref{para: before r}) or failing the Leader-check (\cref{para: after r}), but rather when $b_j^r$ modifies the read value (\cref{para: during r}). 
However, when $b_j^r$ commits before $t$, we can determine its effect on the read value and evaluate STO for $t$. 
This scenario explains why, even when Type $\beta$ transactions are abundant and ``Cross-shard Failure'' rates are high, our consensus latencies remain approximately 25\% lower (\cref{fig:exp2}) than those of Bullshark. 
Even when all transactions are cross-shard, we maintain $\sim$18\% consensus latency reduction when ``Cross-shard Count = 4''  (details in \cref{app: add. results}).

\subsection{Performance under Failures}

\begin{figure}[t]
    \centering
    \setlength{\abovecaptionskip}{1pt}
    \begin{subfigure}{\linewidth}
         \centering
         \includegraphics[width=0.8\linewidth,trim=0cm 0cm 0cm 0cm]{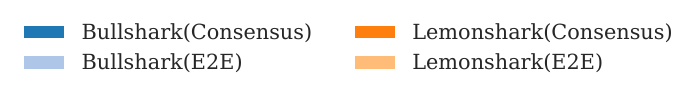}
     \end{subfigure}\\
     \vspace{-0.4\baselineskip}
    \begin{subfigure}[]{0.45\linewidth}
         \centering
         \includegraphics[width=\textwidth,trim=0cm 0cm 0cm .1cm]{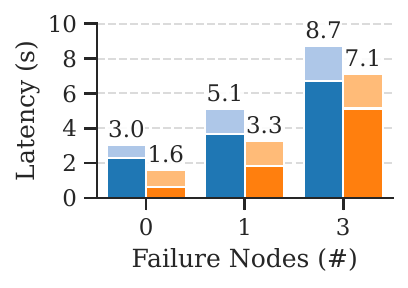}
         \vspace{-20pt}
         \subcaption{}{}
         \label{fig:exp1}
     \end{subfigure}
     \begin{subfigure}[]{0.45\linewidth}
         \centering
         \includegraphics[width=\textwidth,trim=0cm 0cm 0cm .1cm]{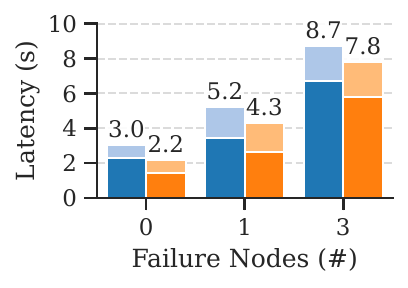}
            \vspace{-20pt}
         \subcaption{}{}
         \label{fig:exp3-2}

     \end{subfigure}
    \caption{Performance of (a) Type $\alpha$ (b) Type $\beta/\gamma$ (with moderate amount of cross-shard activity, (Cross-shard Count = 4, Cross-Shard Failure = 33\%, while varying the number of faults. }
    \vspace{-5pt}
    \label{fig: exp failures}
\end{figure}

We vary the number of crash-faults to evaluate their impact on consensus latency. When $f=1$, \sys's consensus latency improvement decreased from approximately 65\% to approximately 50\%, and further declined to approximately 24\% when $f=3$. As demonstrated in \cref{fig: exp failures}, Type $\beta/\gamma$ transactions exhibited similar performance degradation: approximately 23\% improvement when $f=1$ and approximately 14\% when $f=3$. This behavior is expected, as commitment disruption reduces the likelihood of the scenario described in \cref{para: during r}. For Type $\beta$ transactions, $t$ and $b_j^r$ are often committed together by a fallback leader instead. 
Nevertheless, \sys consistently achieves lower latencies than Bullshark across all failure scenarios.

\subsubsection{Missing Blocks In-Charge of a Shard}
\mich{i was quite unclear here last time, fixed some stuff to add more clarity}
Since only a single node may write to a particular shard per round, if the node in-charge of the shard is faulty, transactions designated to that shard will be delayed till an honest becomes in-charge of that shard. 
This represents a fundamental trade-off in \sys's sharding design; if a transaction is submitted while the node in-charge of it is unavailable, that transaction will incur an additional delay. 
\textbf{Those unfortunate transactions perform slightly worse} than in Bullshark, with an average increase of approximately 500ms (approximately 12\%) in end-to-end latency when $f=1$, and approximately 1500ms (approximately 17\%) when $f=3$.

\section{Related Work}

\paragraph{Traditional Asynchronous BFT}
Traditional DAG-free asynchronous BFT protocols~\cite{honeybadger,Lu2020DumboMVBAOM,Duan2018BEATAB,Liu2020EPICEA} utilize either the BKR~\cite{ben-or,bkr} or CKPS~\cite{Cachin2001SecureAE} agreement-on-common set paradigm; both with expected $O(\log n)$ time complexity. DAG-Rider~\cite{dag-rider} established the current standard for asynchronous BAB protocols by achieving constant expected time complexity with high throughput by establishing a common core\cite{Canetti1995StudiesIS,Dolev2016SomeGI} and subsequently electing a leader.

\paragraph{DAG-BFT with Weaker Network Assumptions}

Several descendants~\cite{bullshark,shoal,shoal++} of early DAG-BFT protocols~\cite{hashgraph,narwhal,dag-rider} reduce commitment latency through leader-focused strategies, such as increasing leader count or introducing faster commitment rules. These improvements come at the cost of reduced resilience to asynchrony~\cite{shoal,shoal++}.
Autobahn~\cite{autobahn} addresses the inefficiency of lock-step dissemination in DAG-BFT but employs PBFT~\cite{pbft}-style commitment, weakening network assumptions to partial synchrony.
In contrast, \sys preserves strong asynchrony assumptions, aligning with DAG-Rider~\cite{dag-rider} and Tusk~\cite{narwhal}. Since \sys builds upon Bullshark, our insights extend to those protocols\cite{dag-rider,narwhal}.

\paragraph{Uncertified DAG-BFT}
Recent work has introduced ``uncertified DAGs''~\cite{cordial,bbca} to reduce DAG-BFT latency by replacing expensive reliable broadcast primitives with best-effort dissemination. However, this introduces a synchronization overhead on the critical path~\cite{autobahn,shoal++}. This overhead can outweigh the latency savings, resulting in significant performance degradation under minor message loss.

\paragraph{Fast Finality in other works}

Other uncertified DAG-BFT protocols, such as Mysticeti-FPC~\cite{Babel2023MysticetiRT}, employ explicit voting to ensure conflict-free transaction ordering for distinct keys, allowing certain transactions to achieve finality faster. However, this approach differs fundamentally from \sys:  Mysticeti-FPC requires explicit coordination, whereas \sys implicitly determines whether early finality conditions are satisfied through local DAG inspection. 

For simplicity, \sys evaluates early finality at the coarse-grained block level. However, the mechanism can be refined to operate at the transaction level exclusively as well (see \cref{app: extending to only STO}). Moreover, Mysticeti-FPC trades network resilience for performance and still inherits the synchronization issues fundamental to all uncertified DAG-BFT protocols.

Similarly, Hotstuff-1~\cite{Kang2024HotStuff1LC} nodes ``vote'' explicitly by conveying their speculative execution results to clients. However, Hotstuff-1 also lacks resilience to network asynchrony, and its single-leader architecture limits throughput scalability.
\section{Conclusion}
In this work, we identify latency disparities between leader and non-leader blocks in DAG-BFT protocols and present \sys: an asynchronous DAG-BFT protocol that enables \textit{early finality}. Early finality allows non-leader blocks to return finalized results before commitment when specific conditions are satisfied, enabling them to achieve optimal latencies previously exclusive to leader blocks. 
\section*{Acknowledgements}
We thank our shepherd Andreas Haeberlen and the anonymous NSDI reviewers for their insightful comments and valuable feedback. We are grateful to Natacha Crooks for her valuable discussions.
Jialin Li was supported by the Singapore Ministry of Education, under Academic Research Fund Tier 2 grant MOE-T2EP20222-0016.

\newpage
{\footnotesize \bibliographystyle{acm}
\bibliography{biblo}}
\appendix

\section{Definitions and Proofs}

\renewcommand{\thealgorithm}{A-\arabic{algorithm}}
\setcounter{algorithm}{0}
\renewcommand{\thefigure}{A-\arabic{figure}}
\setcounter{figure}{0}

\subsection{Definitions for Bullshark}
This subsection describes the Bullshark consensus engine.
We utilize definitions to illustrate better how \sys works in the later sections. 

\subsubsection{Blocks}
\begin{definition}[Rounds and Waves]
\label{def: Rounds and waves}
    The protocol progresses in virtual rounds; For every \textit{round}.
    Beginning from round 1, every 4 rounds constitute a \textit{wave}; round 1-4 belongs to wave 1, round 5-8 to round 2, and so on.
\end{definition}

In each round ($r$), a node $p_i$ may call upon the Reliable Broadcast Primitive to disseminate a message ($rbc(m,p_i,r)$) where $m$ is the message. When a node receives that message, it calls $deliver(m,p_i,r)$. 
The Reliable Broadcast primitive guarantees the following properties:

\begin{itemize}

    \item \textbf{Agreement: }If two honest party calls $deliver(m,p_i,r)$, $deliver(m',p_i,r)$, then $m' = m$. 

    \item \textbf{Validity: }If the sending node is honest, all honest nodes will eventually call $deliver(m,p_i,r)$
    
    \item \textbf{Totality: }If an honest node calls $deliver(m,p_i,r)$, all honest nodes eventually call $deliver(m,p_i,r)$. 
\end{itemize}

In this work, we imagine a 2-phased reliable broadcast primitive akin to Bracha's reliable broadcast \cite{Bracha1987AsynchronousBA}. 

\begin{definition}[Blocks and Pointers]
    \label{def:Blocks and pointers}
    A block $b$ is the result of a message being delivered $deliver(m,p_i,r)$; we say $p_i$ produced the block $b$ in round $r$. 
    \lijl{Any node can deliver the block. ``Produce'' here means the node that does the RBC?}
    \mich{fixed it: removed "delivered by"}
    The message $m$ includes a set of transactions $b = [t_1,\ldots,t_i,\ldots,t_p]$ ordered by $p_i$, as well as a set of at least $2f+1$ blocks from $r-1$. We say $b$ has pointers to any block within that set. 
\end{definition}

In Bullshark, blocks from $r$ might have direct pointers to blocks in $r':r'< r-1$.
These pointers are called ``weak-links''.
However, in this work, we disregard this optimization and focus solely on what they refer to as ``strong-links''; or pointers to the immediate previous round $r-1$.

It is also clear how the blocks and their pointers intrinsically form a DAG, where the blocks are the vertices and the pointers represent the edges. 

\begin{definition} [Block Path]
    For a block $b'$ in round $r'$ and a block $b$ in round $r:r'>r$, $b'$ has a path to $b$ if 1) $b'$ has a pointer to $b$, or 2) there exist blocks in rounds $r'-1$ to $r+1$ that $b'$ has pointers recursively to $b$.
\end{definition}

\subsubsection{Leaders and Votes}

In Bullshark, there exist two classes of leaders: Stable leaders for when the network is relatively synchronous, and fallback leaders when there exists byzantine interference or asynchrony.

\begin{definition}[Stable Leader]
\label{Def:Stable leader}
    A Stable leader is a pseudonym that is deterministically given to a block in the first and third rounds of any wave.
\end{definition}
Stable leaders are typically assigned in a round-robin manner.

\begin{definition}[Fallback Leader]
\label{Def:Fallback leader}
    A fallback leader is a pseudonym randomly given to a block in the first round of any wave.
    This allocation is only revealed at the end of the fourth round of the wave, utilizing a global perfect coin.
\end{definition}

\begin{definition}[Raw Causal history of a block]
    The raw causal history of a particular block $b$ is the blocks to which it has paths to.
\end{definition}

\begin{definition}[Stable Vote]
    For a particular wave $w: w>1$, in the raw causal history of the block produced by $p_i$ in the first round of the wave: if it shows either the 2nd Stable leader in wave $w-1$ or the Fallback leader is committed, then paths by blocks produced in the 2nd and last rounds in $w$ by the same node are considered stable votes.
    \lijl{This definition is not very clear. Consider revising.}
    \mich{Tweaked slightly}
\end{definition}

\begin{definition}[Fallback Vote]
    For a particular wave $w:w>1$. In the raw causal history of the block produced by $p_i$ in the first round of the wave: if it shows neither the 2nd Stable leader in wave $w-1$ nor the Fallback leader is committed, paths by the block produced by $p_i$ in the last round of $w$ to the Fallback leader are considered as a fallback vote.
\end{definition}

In Bullshark, a leader that has garnered sufficient votes is considered \textbf{Committed}. By quorum intersection, it's clear to see that only a single leader type may be committed for the first round of a wave.

\begin{definition}[Committed Leader]
\label{Def:Final Leader}
    A leader is ``Committed'' once it has received sufficient votes; a fallback leader that has received at least $2f+1$ fallback votes, or a stable leader that has received at least $2f+1$ stable votes in the same wave. 
    A leader is also ``Committed'' if it's in another committed leader's causal history, and has obtained sufficient votes: $\geq f+1$ of the appropriate vote type with less than $f+1$ votes of the other vote type present. 
    \lijl{sufficient votes received by the ``other'' committed leader?}.
    \mich{In Bullshark, it is possible that a leader has not obtained 2f+1 votes but in the causal history of a committed leader and received f+1 or more votes (and not enough votes for the other type), it will be ordered before the committed leader. Added slightly more details. }
\end{definition}

\subsubsection{Causal Histories}
Once a leader is committed, 
the blocks it has paths to are ordered deterministically. As such, we define:

\begin{definition}[(Sorted) Causal History of a Block]
\label{def: sorted causal history}
For a block $b$, its causal history is the set of nodes that are part of a sub-DAG, where $b$ is the root, exclusive of the blocks that are in the previous known committed leader's causal history.
Its sorted causal history is obtained by applying Kahn's algorithm to the sub-DAG and reversing the list. where ties (blocks of the same round) are broken deterministically.
We denote this list as $H_{b}=[\ldots,b]$.
\end{definition}

The causal history of any block excludes blocks that are already committed, as those blocks cannot be executed again and are safe to ignore.
In contrast to previous works that accept any deterministic method, we specify a particular ordering procedure. This choice is both intuitive (older blocks execute first, in round-by-round order) and critical to \sys, ensuring that older blocks never precede newer blocks. 
We do not impose constraints on how blocks from the same round are ordered, provided the ordering remains deterministic. 
We illustrate this concept in \cref{fig: reversekhan}.

Once a leader commits, we commit and execute transactions within blocks of its sorted causal history sequentially, proceeding block by block.
Consequently, non-leader blocks are \textit{committed} if and only if they exist in the sorted causal history of a committed leader block, where any non-leader block may be included in at most one committed leader's causal history.
For a committed leader $b'$, we say it \textit{commits} $b:b \in H_{b'}$ as well as all transactions in $b$. 
Given the strict monotonic ordering between committed leaders, for leaders $b'$ and $b''$ that are committed consecutively, all elements in $H_{b'}$ are committed before those in $H_{b''}$.

\begin{definition}[Commitment of a non-leader block]
    For a non-leader block $b$, it is committed if it is in the causal history of a committed leader block $b': b \in H_{b'}$ and $b'$ has sufficient votes. We say $b$ and its transactions are committed by $b'$ in the order of $H_{b'}$.
\end{definition}

\begin{definition}[Commitment ordering]
\label{def: ordering}
    For two subsequently committed leaders $b',b''$ from rounds $r',r''$ such that $r'<r''$ where $H_{b'} = [b_1,b_2,b'],H_{b''}=[b_x,b_y,b'']$ 
    We say $b'$ and elements in $H_{b'}$ are committed before $b''$ and elements in $H_{b''}$. 
    Similarly in $H_{b'}$ we say $b_1$ is committed before $b_2$.
\end{definition}

    

\subsubsection{Transaction/Block Outcomes}

\begin{definition}[Key-spaces and Transactions]
    Let $K$ represent the key-space of a database, and a transaction as a unit of work that modifies a key $k\in K$ in the database. 
\end{definition}

\begin{definition} [Transaction Outcome (\textit{TO})]
    For a block $b = [t_1,\ldots,t_i,\ldots,t_p]$, the transaction outcome (TO) of $t_i$ is the outcome of $t_i$ when executing transactions in the following order: $H_{b}[:-1]+[t_1,\ldots,t_i]$, where $H_{b}[:-1]$ is the prefix of $H_{b}$ exclusive of $b$.
    \lijl{What does ``outcome'' mean here?}
    \mich{outcome is just the result of the transaction. It can be pictured as the resultant state of the keys modified.}
    \lijl{One issue: $H_{b}$ excludes histories of the previous committed leader. It is unclear if that excluded history impacts TO.}
    \mich{Since the previous committed leader's causal history is already ordered and executed. It should not matter in this case. }
\end{definition}

\begin{definition}[Block Outcome \textit{(BO)}]
     For a block $b$ and its Sorted Causal History $H_b$, the Block Outcome (BO) of $b$ is the execution results of all transactions in $b$ after executing the blocks in the order of $H_b$.
     \lijl{Isn't this just the transaction outcome of all transactions in $b$?}
     \mich{yes}
\end{definition}

It is key to note that a block's causal history changes based on which leader a node believes is the latest that has been committed.
Therefore, a node's view on a block's sorted causal history, as well as the TO of the transactions within the block might not be static.

\begin{definition}[Execution Prefix (Block)]
    For blocks $b',b: b=[t_1,\ldots,t_i,\ldots,t_p]$, where $b \in H_{b'}$. 
    The execution prefix $b'\langle b \rangle$ is the outcome of the transactions in $b$ when executing the prefix of $H_{b'}$ up to and including $b$ (\textit{i.e.}, $H_{b'}[0:index(b)]$ or $H_{b'}[0:index(b)-1]+[t_1,\ldots,t_i,\ldots,t_p]$ ). 
    We describe this as the execution prefix of $b$ with respect to $b'$.
    \lijl{It doesn't seem important to include the transactions in $b$ in this definition.}
\end{definition}

\begin{definition} [Execution Prefix (Transactions)]
    For blocks $b',b: b=[t_1,\ldots,t_i,\ldots,t_p]$, where $b \in H_{b'}$. 
    The execution prefix $b'\langle b(t_i) \rangle$ is outcome of transaction $t_i$ when executing the prefix of $H_{b'}$ right before $b$ and $[t_1,\ldots,t_i]$ (\textit{i.e.}, $H_{b'}[0:index(b)-1]+ [t_1,\ldots,t_i]$).
    We describe this as the execution prefix of $t_i$ with respect to $b'$.
\end{definition}

When the leader block $b'$ (where $b \in H_{b'}$) obtains sufficient votes to commit in round $r'$, the execution prefix of each block/transaction in $H_{b'}$ with respect to $b'$ becomes the \emph{finalized} immutable outcome.

\begin{definition}[Safe Transaction Outcome (STO)]
    For a particular transaction $t_i \in b$, we say that its transaction outcome is \textit{safe} if for a committed leader block $b': b\in H_{b'}$, the transaction outcome of $t_i$ is equivalent to the execution prefix $b' \langle b(t_i) \rangle $.
    \lijl{It is important to state the transaction outcome with respect to which block. I assume it is with respect to block $b$}
\end{definition}

\begin{definition}[Safe Transaction Outcome (SBO)]
    If all transactions within a block $b$ have STO, we say the block has a \textit{safe} block outcome (SBO). 
\end{definition}

\begin{definition}[Early Finality]
For a non-leader block $b$, early finality is achieved for $b$ if the SBO of $b$ may be determined before $b$ is determined to be committed. 
\lijl{I believe the commit rule of non-leader blocks was not explicitly defined.}
\lijl{We need to clearly define the timing (local vs global), so terminology such as ``before'' is meaningful.}
\mich{I think its mentioned in Def A.11}
\end{definition}

\begin{definition}[Block Persistence]
\label{def: Persistence}
    We say that a block $b$ in round $r$ persists in round $r':r'>r$, if for any set of $2f+1$ blocks in round $r'$, at least one block in the set has a path to $b$.
    \lijl{Is this from the perspective of a local DAG? Or globally?}
    \mich{Does not matter, it applies to any set of $2f+1$ blocks you have}
\end{definition} 

\begin{proposition}
    \label{prop:persistence}
    A block $b$ in round $r$ has greater than $f$ blocks in round $r+1$ pointing to it \textit{iff} it persists in round $r+1$.
\end{proposition}

\begin{proof}
    By quorum intersection.
    If there are at least $f+1$ blocks in round $r+1$ which have a path to $b$, then clearly any subset of size $2f+1$ blocks must contain at least one such block. 
    If $b$ persists, assume for the sake of contradiction that at most $f$ blocks point to it. Take the set of all blocks that do not have a path to $b$. This set clearly has size at least $2f+1$, contradicting that $b$ persists.
\end{proof}

The notion of Block persistence is that if a block has sufficient pointers to it at some round, then it will definitely be included in some committed leader's causal history.

\subsubsection{Impossibility of Early Finality in Bullshark}
In Bullshark, there is no restriction on which transactions may be included in which block.
As such, a block may include transactions that touch on all keys in $K$. 
Therefore, in our proof, we demonstrate that it's precisely because of this general approach to transaction block inclusion that leads to Bullshark being unable to achieve early finality.

\begin{proposition}
\label{prop: some no persist}
    Due to network asynchrony, there may exist at least $f$ blocks from round $r$ that do not persist in $r+1$. 
\end{proposition}

\begin{proof}
    Suppose at round $r$, there exists a 2 partitions of nodes: $\pi_1,\pi_2$ such that $\pi_1 \cap \pi_2 = \emptyset, \pi_1 \cup \pi_2 = \Pi, |\pi_1|= 2f+1$.
    Its possible that the blocks in $\pi_1$ only learn of blocks in its own partition. As a result, its blocks in $r+1$ only point to the blocks produced by nodes in $\pi_1$. 
    
    Suppose at this point the partition is lifted. The nodes in $\pi_2$ now see all blocks from round $r$. Each block in $\pi_2$ may point to all blocks created by those in $\pi_2$, but those blocks may only get at most $f$ pointers, and therefore do not persist in $r+1$.


\end{proof}

\mich{interesting enough, it seems that with more nodes, the number of blocks that may not persist increases. Suppose $f=1$, we can guarantee at most 1 block does not persist. But with $f=3$ i can have $4$. }
\mich{not really necessary, but I put this here:}

\begin{lemma}
    \label{lemma: Impossibility of early finality in Bullshark}
    Consider the Bullshark protocol; in the presence of Byzantine adversaries and network asynchrony, it is not possible to determine if a non-leader block $b$ in round $r$ has SBO before it is committed. 
\end{lemma}

\begin{proof}
    Suppose $b$ persists in round $r+1$ and there exists a committed leader $b'$ from $r'$ obtains sufficient votes in round $r'+\omega$ such that $b \in H_{b'}$. We focus on a single transaction $t \in b$, that modifies key $k\in K$.

    Consider some block $a$ in round $r$ which contains transactions that modifies every single key in $K$. By~\cref{prop: some no persist}, it is possible that by the means of an inactive adversary, or manipulated network ordering, $a$ does not persist in round $r+1$.

    It is clear that if $a$ is ordered before $b$ in execution, it will affect the BO of $b$. Since $a$ does not persist, for any block $b'$ there is guarantee that $a \in H_{b'}$. Suppose that it holds that $a \in H_{b'}$, and for a block $b''$ it holds that $a \in H_{b''}$.

    For early finality, we must evaluate if $b$ has SBO before it is committed. 
    
    Now, at round $r'+\omega - 1$, it is not known if $b'$ will get sufficient votes to be committed as leader. If it is committed, then to achieve early finality, we must evaluate if $b$ has SBO by this round. Suppose the SBO evaluation is that $a$ is not committed and executed before transactions of $b$. Then consider that at round $r' + \omega$ it is revealed that $b'$ is committed as leader, and so $a$ is ordered before $b$. This conflicts with the evaluated BO for $b$, and violates correctness. Conversely, suppose the SBO evaluation is that $a$ is committed. Then suppose that at round $r'+ \omega$ it is revealed that $b'$ is not committed, and eventually the next leader committed is $b''$, where we have that $a \notin H_{b''}$, again violating correctness. 
    As such, it is not possible to correctly decide on the finalized outcome of executing the transactions in $b$ at round $r'+\omega - 1$. 
    
    Therefore, early finality is not possible in the presence of Byzantine adversaries and network asynchrony


\end{proof}


\subsection{Lemonshark}
\label{app:lemon proofs}

\begin{definition}[Sharded Key-space]
For the Key-space $K= \{k_1,k_2,\ldots,k_n\}$, let it be partitioned into $n$ distinct shards $k_i = \{k_i^1,k_i^2,\ldots\}:\forall i,j \in n,i\neq j, k_i \cap k_j = \emptyset$. We say a block is in-charge of a shard if its transactions only modify keys from that particular shard.
\end{definition}

\begin{definition}[Shard ownership]
   A block $b$ from round $r$ and in-charge of a shard $k_i \in K$ is denoted as $b_i^r$.
\end{definition}

Lemonshark supports three types of transactions, each capable of early finality. For simplicity and as a proof of concept, we present these transaction types and argue that they are sufficient to cover the essential operations in typical database usage, namely, local updates, cross-shard reads, and coordinated atomic actions. These transaction types are:

\begin{itemize}
    \item \textbf{Type $\alpha$: }A transaction in a block in-charge of $k_i$ that reads from and writes to $k_i$. 

    \item \textbf{Type $\beta$: }A transaction in a block in-charge of $k_i$ that reads from a key in $k_j$ where $j\neq i$ and writes to $k_i$.

    \mich{changed to sub-transactions}
    \item \textbf{Type $\gamma$: }An unordered pair of Type $\alpha/\beta$ sub-transactions that are atomic and pair-wise serializable. 
\end{itemize}

For clarity, we focus on Type $\beta$ transactions that read from a single other shard, and Type $\gamma$ transactions that are sub-transaction pairs. Extensions to Type $\beta$ transactions reading from arbitrary numbers of shards and Type $\gamma$ transactions as $n$-tuples are detailed in \cref{app: extending beta gamma}.

\begin{definition}[Pair-wise serializable]
    A pair of sub-transactions is considered Pair-wise serializable if it is executed concurrently, with no other transaction interleaving it. 
\end{definition}

Since a Type $\gamma$ transaction is two separate sub-transactions that may exist in two separate blocks, it's possible that one sub-transaction is committed before the other. For Type $\gamma$ sub-transactions to be pair-wise serializable, they must be executed together. Therefore, the earlier committed sub-transaction must be delayed. This is achieved by means of a \textit{Delay List}.

\begin{definition}[Delay List ($DL$)]
    A Delay list from round $r$ includes an ordered list of transactions belonging to rounds up to $r$; we denote this at $DL_r$. Any transaction $t$ that reads or modifies a key from round $r$ automatically fails to gain STO if there exists a transaction in $DL_r$ that also modifies the same key.  
    
\end{definition}
\mich{proof sketches from here on out, will polish}
Suppose for a pair of sub-transactions $t_1,t_2$ that make up a Type $\gamma$ transaction. If one of the sub-transactions $(t_1)$ is committed before the other, or exists in an earlier round than the other $(t_2)$, the former is placed into the ordered delay list.

It is only removed from the DL once $t_2$ is committed or has STO. 
$DL_r$ includes all transactions that might be included in it before and up to round $r$.
Since a transaction that is placed into the delay list has an unknown outcome until the other half is committed, we make the restriction that a transaction $t$ that modifies or reads a key $k_i^j$ may not have STO if there exists a transaction in $DL$ that also modifies $k_i^j$. 

Since only one block may modify a key $k$ every round, we need not worry about concurrent additions into the $DL_r$ involving the same shard. 

\subsection{Leader check}
Here, we discuss in more formal detail the leader check mentioned in ~\cref{para: leader exp}
\begin{figure*}[t]
    \centering
    \includegraphics[width=0.8\textwidth, trim=0cm 0cm 0cm -.5cm]{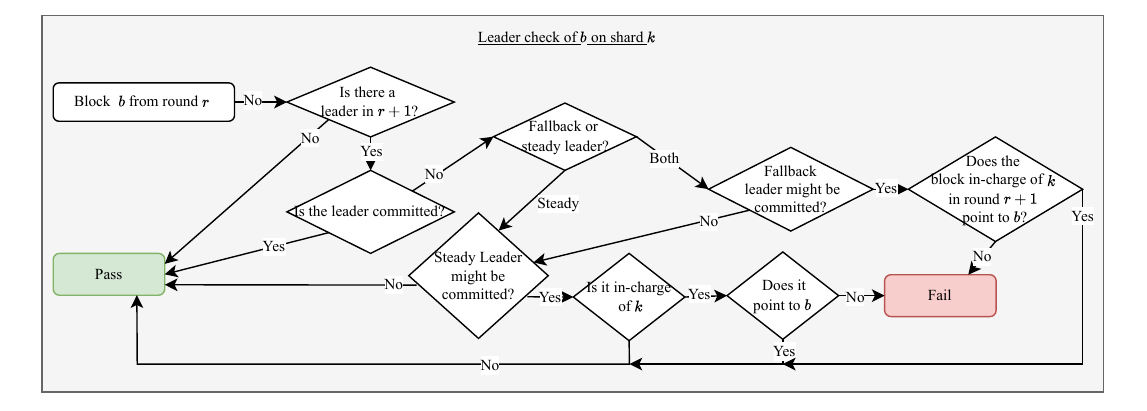}
    \caption{This Figure illustrates the checks to ensure that if a leader block in-charge of a shard $k_i$ in the immediate subsequent round exists, it must be executed after block $b$.}
    \label{fig: leadercheck}
\end{figure*}

    
    
    


\begin{algorithm}
\small
\caption{Leader Check}
\label{alg:leader_check}
\begin{algorithmic}[1]
\Function{Leader\_Check}{$b^r$, $k_i$}
    \If{leader in round $r+1$ \textbf{or} leader in round $r+1$ is committed}
        \State \Return \textbf{true}
    \EndIf
    
    \State $w \leftarrow \lfloor r/4 \rfloor + 1$
    \State $\text{steady\_ok} \leftarrow$ enough steady votes in wave $w$
    \State $\text{fallback\_ok} \leftarrow$ enough fallback votes in wave $w$
    
    \If{\textbf{not} steady\_ok \textbf{and not} fallback\_ok}
        \State \Return \textbf{true}
    \EndIf
    
    \State $\text{steady\_leader\_check} \leftarrow$ steady leader in $r+1$ is in charge of $k_i$
    
    \If{fallback\_ok \textbf{and} both leaders exist}
        \State \Return \Call{CheckPath}{$b^r$, $b^{r+1}_i$}
    \EndIf
    
    \If{steady\_ok \textbf{and} steady\_leader\_check}
        \State \Return \Call{CheckPath}{$b^r$, $b^{r+1}_i$}
    \EndIf
    
    \State \Return \textbf{true}
\EndFunction

\Function{CheckPath}{$b^r$, $b^{r+1}_i$}
    \State \Return $b^{r+1}_i$ has path to $b^r$
\EndFunction
\end{algorithmic}
\end{algorithm}

\begin{definition} [Leader Check]
    For a block $b$ in round $r$, it $b$ passes the leader check on shard $k_i$ if:

    \begin{enumerate}
        \item In the next round, there exist no leaders (odd round). 

        \item if there exists a steady leader and there exists sufficient steady votes, if it is in-charge of $k_i$, it must have a pointer to $b$. 

        \item If there exists sufficient fallback votes, $b_i^{r+1}$ must have a pointer to $b$.

    \end{enumerate}

    Or a leader in $r+1$ is already committed, while $b$ is not. 
    Otherwise, the block fails the leader check on shard $k_i$. 
\end{definition}

Note that $b$ does not need to be in-charge of shard $k_i$ in this definition.

We now prove using the following propositions that if the leader check is satisfied for a block $b$, then we have sufficient conditions to be certain that the block in-charge of $k_i$ in round $r+1$ cannot execute before $b$.

\mich{added \cref{def: ordering} to aid this abit. I also broke the original one into two. }
\begin{proposition}
    \label{prop: leader check_1}
    Let $b$ be an uncommitted block in round $r$ that persists in round $r+1$. If $b$ passes the leader check for shard $k_i$, then $b$ will always be committed before $b_i^{r+1}$.


       

\end{proposition}

\begin{proof}

    Suppose otherwise, $b_i^{r+1}$ executes before $b$. 
    
    Consider the case that $b_i^{r+1}$ is committed due to a committed leader $b'$ in round $r+2$ or after. As $b$ persists in round $r+1$, $b'$ must have a path to $b$ and so it will also commit $b$. Recall from \cref{def: sorted causal history} that blocks are ordered in non-descending order of the rounds they are created. Therefore, $b$ must be ordered to execute before $b_i^{r+1}$ in $H_{b'}$ and this case is impossible.

    Therefore, the committed leader $b'$ must exist in round $r+1$. Since $b_i^{r+1}$ is committed by a leader block in round $r+1$, this implies $b_i^{r+1} = b'$. As $b$ passes the leader check for shard $k$, the second and third items of the leader check condition guarantee that $b'$ must have a pointer to $b$. However, this means $b$ will be committed by $b_i^{r+1}$ and by the ordering of blocks in the causal history, $b$ will execute before $b_i^{r+1}$.

\end{proof}
This shows that the only block from $r+1$ that can possibly be committed before a block from $r$ that persists in $r+1$ is a leader block in $r+1$. However, if this leader has pointers to the persisting block, it will \textit{not} be committed before the persisting block due to the ordering defined in \cref{def: sorted causal history}.

\begin{proposition}
\label{prop: leader check_2}
    Let $b$ be an uncommitted block in round $r$ that persists in round $r+1$. If there exists a leader $b'$ in round $r+1$ where $b \notin H_{b'}$, and $b'$ is already known to be committed, then any uncommitted block in $r+1$ may not be executed before $b$. 
\end{proposition}

\begin{proof}
    As $b$ is not in $H_{b'}$, the next elected leader that can commit $b$ must be in round $r+2$ or after. Since $b$ persists in $r+1$, this block must contain $b$ in its causal history and so commit $b$. Due to the ordering defined in \cref{def: sorted causal history}, no blocks in round $r+1$ may be ordered and execute before $b$ even if they exist and are in the next leader's causal history. 
\end{proof}



\subsubsection{STO for Type $\alpha$}

\mich{@ alvin, added this to help with the proofs below}

Let us define a notion that will be useful when showing that we can correctly determine STO for a transaction before the parent block is committed.

\begin{definition}[Complete Shard History]
\label{def: complete shard history}
    For an uncommitted block $b$ in round $r$ and a shard $k_i$, let $b_i^{\hat{r}}$ be the earliest uncommitted block in charge of shard $k_i$ from round $\hat{r}$ where $\hat{r} \leq r$.
    
    We say $b$ has \textit{Complete Shard History} for $k_i$ if and only if $b$ has a path to $\hat{b}$ which includes all blocks from round $r-1$ to $\hat{r}+1$ that are in-charge of $k_i$.
     We denote this ordered path of blocks as $C_b(i) = [b_i^{\hat{r}},b_i^{\hat{r}+1} \ldots,b_i^{r-1}, b]$, where $|C_b(i)| = r-\hat{r}+1$.
\end{definition}

Observe that by requiring a path containing all the blocks in-charge of $k_i$, each block on this path also has Complete Shard History. Note that $b$ does not have to be in-charge of $k$ to have Complete Shard History for $k_i$. 

Now, we show why Complete Shard History is useful in helping correctly determine STO. Suppose some block $b$ in round $r$ has Complete Shard History for shard $k_i$. Then $b$ knows of all of the relevant transactions on shard $k_i$ that can be committed. Further, if the next committed leader has a pointer to $b$, then the leader must also have pointers to all of these blocks, and so these blocks will be committed. Therefore, $b$ knows that all of these blocks must be executed, and so can correctly determine the outcome of executing these blocks before its own execution.

\begin{proposition}
    \label{prop: complete shard history}
    Consider a block $b$ in round $r$ which has complete shard history for $k_i$.
    If $b$ passes the leader check on $k$ and persists in $r+1$, there may not be any block in-charge of $k_i$ which is not in $H_{b}$ that may be executed before $b$ besides $b^r_i$.

\end{proposition}

\begin{proof}

    Let $b'$ be the leader that eventually commits $b$. Observe that as $b$ has complete shard history for $k_i$, blocks in-charge of $k_i$ but not in $H_{b}$ must be in round $r$ or after. Now, assume otherwise that some block besides $b^r_i$ in-charge of $k_i$ from some round before $r$ is executed before $b$. This implies that it must be committed by some leader $b''$ which commits before $b'$. If it is committed by $b'$, our ordering (see~\cref{def: sorted causal history of a block}) guarantees that it is executed after $b$.

    Suppose $b''$ is in round $r+2$ or after. Since $b$ persists in round $r+1$, it must hold that $b \in H_{b''}$ and so is committed by $b''$. This directly contradicts that $b'$ is the leader that eventually commits $b$.
    
    Now, suppose that $b''$ is in round $r+1$. 
    Since $b$ passes the leader-check, $b''$ must have a pointer to $b$ and therefore commits $b$ and again directly contradicts that $b$ is committed by $b'$.
    
    Finally, if $b''$ is in round $r$, the only block in-charge of $k$ that it can commit but is not in $H_{b}$ is $b^r_i$ (that is, $b''$ is $b^r_i$). It cannot possibly commit any block in rounds $r+1$ or after.
\end{proof}

Note that $b$ does not need to be in-charge of $k_i$ to have complete shard history for shard $k_i$.

\begin{figure}[t]
    \centering
    \includegraphics[width=.8\linewidth, trim=0cm 0cm 0cm -.5cm]{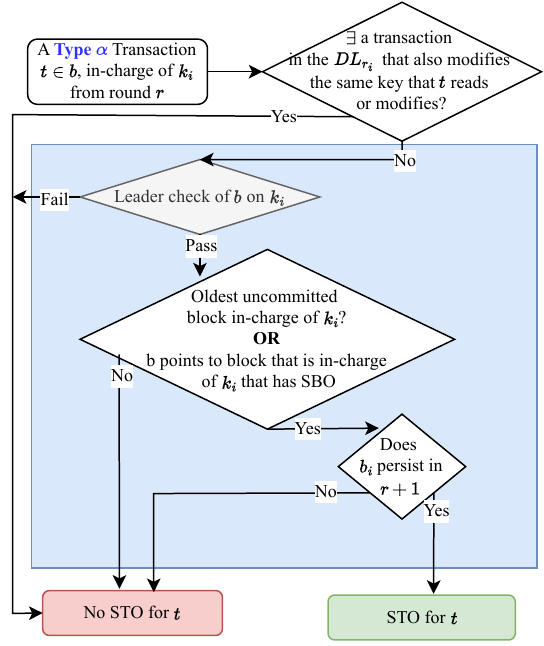}
    \caption{This Figure illustrates all the sufficient conditions for a Type $\alpha$ transaction to have STO.}
    \label{fig: Complete a}
\end{figure}

\begin{lemma}
    \label{lemma: STO for alpha}
    For a shard $k_i$, a Type $\alpha$ transaction $t \in b_i^r$ can be determined to have STO in round $r+1$ if all of the following hold: 
    \begin{itemize}

        \item There does not exist any transaction in $DL_r$ that modifies $k_i$. 
        \item $b_i^{r-1}$ has SBO and $b_i^r$ has a pointer to it, or $b_i^r$ is the earliest uncommitted block in-charge of $k_i$. 
        \item $b_i^r$ passes the leader check on $k_i$.  
        \item $b_i^r$ persists in round $r+1$. 
    \end{itemize}
\end{lemma}

\begin{proof}

    Recall that a transaction has STO if its TO matches the execution prefix of that transaction with respect to the leader that eventually commits it. Let $b'$ be the leader that eventually commits $b_i^r$.
    
    First consider the case where $b_i^r$ is the earliest block in-charge of $k_i$ which is not committed. As $b_i^r$ passes the leader check, by~\cref{prop: leader check_1} it holds that $b_i^{r+1}$ cannot execute before it. Next, as $b_i^r$ persists in round $r+1$, any block in-charge of $k_i$ in round $r+2$ or after cannot execute before it, as any leader committing such a block necessarily has a path to $b_i^r$ and so commits it. By the ordering on the blocks, $b_i^r$ must execute before any block in-charge of $k_i$ in round $r+2$ or later. Therefore, it is safe to conclude that $b_i^r$ will execute before any other blocks in-charge of $k_i$ in rounds after $r$. Therefore, the transaction outcome of $t$ must be equivalent against the execution prefix $b'\langle b(t_i)\rangle$.

    Suppose $b_i^r$ is not the earliest uncommitted block in-charge of $k_i$, but it has a pointer to $b_i^{r-1}$ which has SBO. This implies either $b_i^{r-1}$ is the earliest uncommitted block in-charge of shard $k$ and the earlier argument holds for all its transactions; or it has a pointer to the block that is in-charge of $k_i$ in the preceding round, which has SBO. By applying this argument recursively, this implies that $b_i^r$ has complete shard history of $k_i$. As $b_i^r$ passes the leader check, by~\cref{prop: leader check_1} and~\cref{prop: complete shard history}, we can be sure that there does not exist any block not in $H_{b}$ in-charge of $k_i$ that may be executed before $b_i^r$.

    Since $b_i^{r-1}$ has SBO, its block outcome is equivalent to its execution prefix with respect to the leader block that commits it. As it must also have complete shard history for $k_i$ (at round $r-1$), its block outcome must execute all of the blocks in-charge of $k_i$ in rounds prior to $r-1$ in order of their rounds. Therefore this also holds true for the execution order of blocks in-charge of $k_i$ by committed leader blocks. So it is clear that when committed by $b'$, $b_i^r$ must be the next block in-charge of $k_i$ to be executed after $b_i^{r-1}$. This naturally implies that the transaction outcome of any Type $\alpha$ transaction in $b_i^r$ must be equivalent to the execution prefix of it against $b'$.
    
    Finally, it is straightforward to observe that these conditions can be determined in round $r+1$.

\end{proof}

\begin{proposition}
\label{prop: persist bound}
      Even under network asynchrony and the presence of Byzantine adversaries, there must exist at least $(3f+2)/2$ blocks from round $r$ that persist in $r+1$. 
\end{proposition}
\begin{proof}
    Suppose there exist $3f+1$ blocks in round $r$, and due to byzantine inaction, there exist only $2f+1$ blocks in round $r+1$. There may exist fewer than $3f+1$ blocks in round $r$, but with fewer blocks, more blocks persist, as there are fewer options for blocks in round $r$ to point to. 
    To minimize the number of blocks that persist in round $r+1$, each block in $r+1$ have exactly $2f+1$ pointers to blocks in $r$. 
    
    Suppose there exist $A$ blocks from round $r$ that must persist in. To minimize the pointers from blocks in round $r$ that point to the remainder $3f+1-A$ blocks, we let all blocks in $r+1$ point to those blocks in $A$. 

    This means each block in $r+1$ need to point to a remainder of $2f+1-A$ blocks in round $r$. 

    Since each $2f+1-A$ blocks in the remainder $3f+1-A$ blocks in $r$ can only have at most $f$ blocks pointing to it and there exists $2f+1$ blocks, we need $\lceil 2f+1/f \rceil=3$ of these sets in the $3f+1-A$ blocks. 
    Therefore, we solve the inequality:

    \begin{align*}
        (3f+1-A)(2f+1-A) >= 3 \\
        2f+1 > A >= (3f+2)/2
    \end{align*}
    Therefore, $(3f+2)/2$ is the minimum number of blocks that must persist in $r+1$. 
    
\end{proof}

\cref{prop: persist bound} shows that a certain number of blocks \textbf{must} persist in each round. Therefore, its possible for certain blocks to gain SBO even in the presence of byzantine adversaries and network asynchrony.


\subsubsection{STO for Type $\beta$}

\begin{figure}[t]
    \centering
    \includegraphics[width=0.8\linewidth, trim=0cm 0cm 0cm -.5cm]{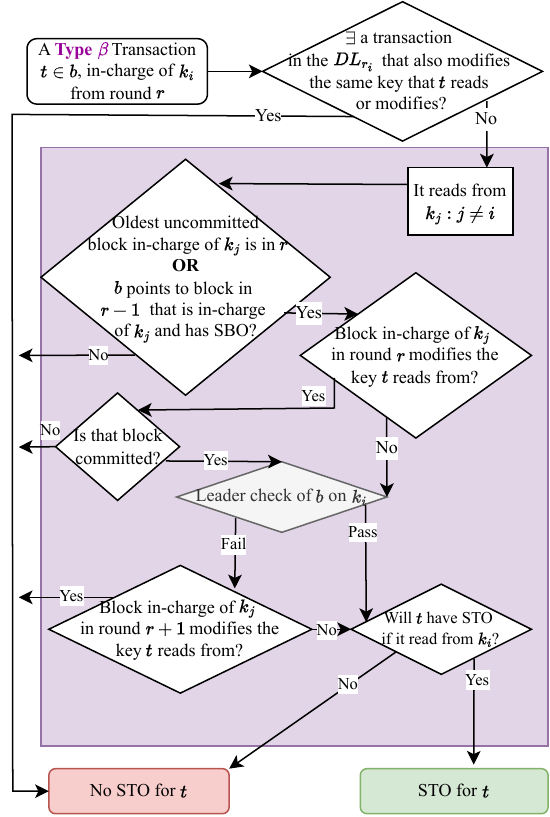}
    \caption{This Figure illustrates all the sufficient conditions for a Type $\beta$ transaction to have STO.}
    \label{fig: Complete b}
\end{figure}

\begin{lemma}
    \label{lemma: STO for beta}
     Take any two distinct shards $k_i, k_j$. Let $b_{i}^r$ be a block in round $r$ in-charge of shard $k_i$, and let $b_{j}^{r-1},b_{j}^r,b_{j}^{r+1}$ be blocks in-charge of $k_j$ in rounds $r-1, r, r+1$ respectively. 
     \mich{changed $b$ -> $b_{i}^r$}
     A Type $\beta$ transaction $t \in b_{i}^r$ that reads from $k_j^a \in k_j$ can be determined to have STO in $r+1$ if all of the following hold: 

     \begin{enumerate}

         \item $t$ meets all of the conditions for STO for a Type $\alpha$ transaction (see~\cref{lemma: STO for alpha}).
         \item $b_{i}^r$ has a pointer to $b_{j}^{r-1}$ and it has SBO, or $b_{j}^{r}$ is the oldest uncommitted block in-charge of $k_j$. 
         \item $b_{j}^r$ does not have a transaction that modifies $k_j^a$, or $b_j^r$ is known to already be committed by some leader block.
         

         \item $b_{i}^r$ passes the leader check for $k_j$ or the $b_{j}^{r+1}$ does not have a transaction that modifies $k_j^a$.
     \end{enumerate}
\end{lemma}

\begin{proof}

    Let $b'$ be the leader block that eventually commits $b_i^r$.
    As the transaction writes to $k_i$, we require that $t$ meets the conditions for STO as if it is a Type $\alpha$ transaction, including that it persists in round $r+1$.
    
    Now, since $t$ reads from $k_j$, we now additionally need to ensure that our read value based on the block outcome of $b_{i}^r$ will be the same as that for $t$ when it is committed, so that the transaction outcome of $t$ is equivalent to the execution prefix of $t$ in $b'$. 

    First, suppose that either $b_{j}^r$ does not have a transaction that modifies $k_j^a$. Given this, executing $b_{j}^r$ before $b_{i}^r$ in leads to an identical execution prefix of $t$ with respect to $b'$ as if it were executed after. Alternatively, if this does not hold, then we know that $b_j^r$ is already committed by some previous leader. If so, then clearly $b_j^{r}$ cannot be in $H_{b'}$ (by definition of causal history of a block) and therefore it will not affect the execution prefix of $t$ with respect to $b'$. In either case, we can thus safely ignore $b_j^r$ as it has no effect on the execution prefix of $t$ with respect to $b'$.

    We have that either $b_{i}^r$ has a pointer to $b_j^{r-1}$ and $b_j^{r-1}$ has SBO, or that $b_j^r$ is the oldest uncommitted block in-charge of $k_j$. We then apply a similar argument as~\cref{lemma: STO for alpha}. This implies that $b$ has complete shard history of $k_j$, and further that the blocks in-charge of $k_j$ in $H_{b_j^{r-1}}$ must be executed in order of their rounds, followed by $b_j^{r-1}$. Therefore, our read value of $k_j^a$ in the transaction outcome of $t$ will be equivalent to the execution prefix of $t$ with respect to $b'$, as long as no block in-charge of $k_j$ after round $r$ affects the execution prefix of $t$ with respect to $b'$.
    
    If $b_i^r$ passes the leader check on $k_j$, by~\cref{prop: complete shard history}, there cannot exist any block in-charge of $k_j$ that can be ordered before $b_{i}^r$ in the next committed leader's causal history besides $b_j^r$. Otherwise, if $b_{i}^r$ does not pass the leader check on $k_j$ then we have that $b_j^{r+1}$ does not have a transaction that modifies $t$. Then it is clear that executing $b_j^{r+1}$ before $b_i^r$ leads to an equivalent execution prefix of $t$ with respect to $b'$ as if it were executed after. Therefore, $b_j^{r+1}$ will not affect the execution prefix of $t$ with respect to $b'$. As $b'$ persists in round $r+1$, then no blocks in-charge of $k_j$ will execute before $b_i^r$. 

    Therefore, we can conclude that the transaction order of $t$ is equivalent to the execution prefix of $t$ with respect to $b'$.

    
    These conditions can be evaluated in round $r+1$.
    
\end{proof}


\subsubsection{STO for Type $\gamma$}

The following proposition shows how we may infer that two blocks may exist in the causal history of a single leader block before even witnessing the said leader block. 

\begin{proposition}
    \label{prop: 2 blocks 1 leader}
    Consider a block $b$ in round $r$ and a block $\hat{b}$ in round $\hat{r}$, and without loss of generality let $r < \hat{r}$. If both blocks persist in round $\hat{r}+1$ and neither block is in the causal history of any leader block before and up to round $\hat{r}+1$, then both blocks must exist in the same committed leader block's causal history. 
\end{proposition}
\begin{proof}
    Since neither block exists in the causal history of leader blocks from before and up to round $\hat{r}+1$, the next possible leader block can contain both must exist after $\hat{r}+1$. Since both blocks persist in round $\hat{r}+1$, that leader block must have a path to both blocks. Therefore, they must exist in the same committed leader block's causal history. 
\end{proof}

The following proposition shows that we will not incorrectly determine a transaction to have STO when there is uncertainty due to a conflicting transaction in the delay list.

\begin{proposition}
    \label{prop: DL}
    Consider two blocks $b_i^{r-1},b_i^r$ in-charge of shard $k_i$ in rounds $r-1, r$ respectively. 
    
    If $b$ points to $b_i^{r-1}$, which has SBO, or if we know $b_i^r$ is the oldest uncommitted block in-charge of $k_i$, then it implies that $DL_r$ contains all possible transactions that modified keys in $k_i$.
 
\end{proposition}

\begin{proof}    
    If $b_i^r$ was the oldest uncommitted block in-charge of $k_i$ and it is in $r$, then any other possible block in-charge of $k_i$ prior to $r$ must have been committed and therefore known. 

    Similarly, if $b_i^{r-1}$ has SBO, it implies that $b$ has complete shard history of $k_i$. Therefore, there does not exist an uncommitted block in-charge of $k_i$ from before round $r$ that is not in $H_b$. Therefore, all blocks in-charge of $k_i$ prior to $r$ is known as well.

    Since all blocks that are in-charge of $k_i$ is known, there must not exist a transaction that would be in $DL_r$ that we do not know. 
\end{proof} 

\mich{added this to help make the ordering and execution of Type $\gamma$ transactions slightly clearer}
\begin{definition}[Type $\gamma$ sub-transactions execution order]
\label{def: gamma execution order}

Consider a Type $\gamma$ transaction comprised of two sub-transactions $t_1,t_2$ that exist in blocks $b_1,b_2$ from rounds $r_1,r_2$ respectively. 

\begin{itemize}
    \item If $r_1 = r_2$, and $b_1,b_2$ are in the same leader's causal history, then $t_1,t_2$ are executed in some deterministic order between the 2 transactions. 

    \item If $r_1 <r_2$, and $b_1,b_2$ are in the same leader's causal history. Then $t_1$ will be executed concurrently with $t_2$. That is to say, it is not executed along with the other transactions in $b_1$. 

    \item If $b_1,b_2$ are committed by different leaders. The one committed earlier executes at $max(r_1,r_2)$ with the one committed later. 
\end{itemize}
\end{definition}

\mich{Type gamma: same round, same leader}
\mich{below is a bit iffy, polishing atm, not sure how to phrase is better. }
\begin{lemma}
    \label{lemma: gamma 1}
    Consider a Type $\gamma$ transaction of 2 sub-transactions $t = [t_1, t_2]$ where $t_1 \in b_1,t_2 \in b_2$ in round $r$.
    The transaction $t$ can be determined to have STO in $r_1+1$ if all of the following hold: 

    \begin{itemize}
        \item Both $t_1$ and $t_2$ may be evaluated to have STO independently. 

        \item Every other transaction in $b_1$ and $b_2$ have STO. 
        
        \item The conditions of~\cref{prop: 2 blocks 1 leader} is met and so there exists some eventually committed leader block $b'$ such that both $t_1, t_2$ exist in $H_{b'}$. 

    \end{itemize}
\end{lemma}
\begin{proof}
    First, given that the conditions of~\cref{prop: 2 blocks 1 leader} hold, we now assume $b_1$ and $b_2$ are committed by the some leader block $b'$, and without loss of generality, let the prime transaction be $t_2$. Given that blocks of the same round are ordered deterministically, it is known that all blocks will order $b_1$ before $b_2$ and this assumption can be safely made.

    Now, given that both $t_1$ and $t_2$ satisfy the requirements to have STO independently, the transaction outcome of $t_1$ and $t_2$ are each equivalent to their respective execution prefix with respect to $b'$. However, now $t_1$ will be concurrently executed with $t_2$ instead of following the typical ordering on the blocks and transactions.
    
    Given that all other transactions in $b_1$ and the transactions prior to $t_2$ in $b_2$ will be guaranteed to execute \emph{before} $t_1$, their execution may now affect the transaction outcome of $t_1$. Given that all of these transactions have STO, their transaction outcome is equivalent to their execution prefix with respect to $b'$ and therefore, it can be correctly determined their effect on transaction outcome of $t_1$. Therefore transaction outcome of $t_1$ now, when ordered concurrently with $t_2$ is indeed equivalent to its execution prefix with respect to $b'$. This is illustrated with \cref{fig: type c1}.

    All of the above may be evaluated in round $r+1$. 

\end{proof}

\mich{type gamma, same leader, differ round. }
\begin{lemma}
\label{lemma: gamma 2}
     Consider a Type $\gamma$ transaction of 2 sub-transactions $t = [t_1, t_2]$ where $t_1 \in b_i^{r_1}$ and $t_2 \in b_j^{r_2}$, and without loss of generality, $r_1 < r_2$.
     
     $t = [t_1,t_2]$ can be determined to have STO in round $r_2+1$ if all of the following hold: 

    \begin{itemize}
        \item $t_1$ independently has STO when considering it to be in $b_i^{r_2}$. This condition includes that $b_i^{r_2}$ must be observed.
        
        \item $t_2$ independently has STO.

        \item Every other transaction in $b_j^{r_2}$ has STO. 

        \item Furthermore, all other transactions in $b_{i}^{r_2}$ have STO given that $t_1$ were to be in $b_{i}^{r_2}$ instead of $b_i^{r_1}$.
        
        \item The conditions of~\cref{prop: 2 blocks 1 leader} is met and so there exists some eventually committed leader block $b'$ such that both $t_1, t_2$ exist in $H_{b'}$. 

    \end{itemize}
\end{lemma}

\begin{proof}
    First, given that the conditions of~\cref{prop: 2 blocks 1 leader} hold, we now assume that $b_i^{r_1}$ and $b_{j}^{r_2}$ are committed by the same leader block $b'$.

    Since $r_1 < r_2$, $t_1$ will be moved to execute concurrently with $t_2$ instead of in $b_i^{r_1}$. Given that $b_i^{r_2}$ is observed, then now let $t$ be removed from $b_i^{r_1}$ and placed as the last transaction in $b_i^{r_2}$ instead. Now, this becomes an instance which is covered by~\cref{lemma: gamma 1}. Given that the conditions of this lemma are met, then it holds that the conditions of~\cref{lemma: gamma 1} are satisfied as well. Therefore, the transaction order of $t_1$ will be equivalent to its execution prefix with respect to $b'$. 

    There is one technical detail that is different, since it is known that $t_1$ will now be executed with $t_2$ (as if in $b_{j}^{r_2}$). Therefore, its presence in the delay list can now be safely ignored when considering STO for transactions in blocks in-charge of $k_i$ in between and including rounds $r_1$ and $r_2$, as it is guaranteed that this transaction will only execute after these blocks.
    
    

    If the above conditions are true, STO may be evaluated at round $r+1$.
\end{proof}

\begin{lemma}
\label{lemma: early finality for all}
    Type $\alpha, \beta, \gamma$ transactions may be evaluated to have STO before the leader $b'$ that includes them in its causal history is committed. 
\end{lemma}

\begin{proof}
    For a Type $\alpha,\beta$ transaction in round $r+1$, it may be evaluated to have STO in round $r+1$ as per \cref{lemma: STO for alpha,lemma: STO for beta}. Since the leader checks were passed, the earliest the transaction may be committed is round $r+2$, where $b'$ is a steady leader. 

    For Type $\gamma$ transactions, that exist in two parts $b_1,b_2$ from rounds $r_1,r_2$ respectively. STO may be evaluated at round $max(r_1,r_2)+1$. whilst the leader may exist in the earliest round $max(r_1,r_2)+1+1$ and committed in round $max(r_1,r_2)+1+2$. 
\end{proof}

Therefore, by \cref{lemma: early finality for all}, early finality is possible for Type $\alpha, \beta, \gamma$ transactions.


\section{Extending Type $\beta/\gamma$ transactions}
\label{app: extending beta gamma}

Extending Type $\beta$ transactions to read from a set of shards is quite simple. Suppose we have a transaction $t \in b_i^r$ that reads from a set of shards $\mathcal{K}=\{\ldots\}$. $\forall k_j \in \mathcal{K}$, $b_i^r$ needs points to a block $b_j^{r-1}$ that has SBO, and it passes the leader check on shard $k_j$. If these conditions are achieved, it's clear to see how \cref{lemma: STO for beta} may be extended to prove that $t$ has STO.

Extending Type $\gamma$ transactions to be an $n$-tuple across $n$ shards is quite simple as well. Suppose all $n$ sub-transactions exist in the same round, have SBOs, and we can ensure that all sub-transactions will eventually be in the causal history of the same committed leader. 
Per \cref{lemma: gamma 1}, the Type $\gamma$ transaction has STO as well. A similar argument may be applied to \cref{lemma: gamma 2} when all $n$ sub-transactions exist in different rounds, but is eventually committed by the same leader. 

\section{Future Work: A Finer Grained \sys}
\label{app: extending to only STO}

For clarity, \sys presents \textit{sufficient} conditions for determining whether a block has SBO, evaluated recursively starting from the earliest uncommitted block in each shard. Intuitively, SBO is ``inherited'' from one block to the next: a block cannot have SBO if the uncommitted block in-charge of the same shard in the previous round does not. 

However, this inheritance requirement is stronger than necessary for individual transactions. Consider blocks $b_1, b_2$---the first and second uncommitted blocks in charge of a shard---and a type $\alpha$ transaction $t \in b_2$ that modifies a key untouched by any transaction in $b_1$. Transaction $t$ can achieve STO as long as $b_2$ persists and passes the leader check, regardless of whether $b_1$ has SBO or whether $b_2$ references it.  

This observation reveals that \sys's current conditions are sufficient but not necessary. Deriving tight \textit{necessary} conditions for STO evaluation remains an interesting direction for future work.

\section{Missing blocks, Weak links, Orphaned and Dangling Blocks}
\label{app: orphaned}
\mich{maybe run this section through llm to fix grammar etc}

Determining the ``\textit{oldest uncommitted block}'' for a shard in the presence of Byzantine nodes requires identifying whether missing blocks are genuinely absent or exist without a node's knowledge. 

Missing blocks can be ascertained by a node querying the remaining nodes to determine whether they voted in the second phase (\texttt{vote} phase) of the reliable broadcast~\cite{Bracha1987AsynchronousBA}. 
If fewer than $2f+1$ nodes voted (ascertained by having $<f+1$ positive responses out of the $2f+1$ responses obtained), such a block will never exist and can be categorized as missing. 

However, if at least $2f+1$ ($\geq f+1$ out of the $2f+1$ responses obtained) nodes voted, that block might exist. This is not problematic if all blocks in subsequent rounds are known and none reference the missing block. We refer to those as orphaned blocks. 
Dangling blocks (blocks pointed to by only a small subset of blocks and never committed) are categorized similarly. Bullshark~\cite{bullshark} allows blocks to reference blocks from non-immediate previous rounds via pointers referred to as ``weak links.''

These weak links do not participate in consensus; instead, they are used to help orphaned or dangling blocks eventually get committed. However, since blocks may be referenced via weak links almost arbitrarily, we disallow such links in \sys, as they enable arbitrary inclusion of blocks into a node's causal history. 

The problem with dangling blocks is that if they are not committed, they will always be the \textit{oldest uncommitted block} in-charge of a particular shard, preventing other blocks in-charge of the same shard from ever gaining SBO. 

In this section, we will define \textit{limited look-back} as a potential solution to handle such cases, effectively acting as a high-water mark that eventually excludes these dangling blocks from consideration.

We first define Sorted Causal History with Limited Look-back to be used in-place of Sorted Causal History. 
\begin{definition}[Sorted Causal History with Limited look-back ($v$)]
\label{def app: lh}
    Suppose the last known committed leader is in round $r'$, where the next possibly committed leader is in round $r'+2$. Consider a block $b$ in round $r>r'$ and let $v$ be a publicly known constant. $b$'s Sorted Causal History with Limited look-back $LH_{b}$ is defined to be all blocks in $H_{b}$ that are from round $r'+2-v$ or after. We denote $r'+2-v$ as the watermark $m_b$ for $b$. 
\end{definition}

Consider some block $b$ in round $r$ and the leader block $b'$ in round $r'$ that eventually commits $b$. One issue that can potentially occur with the above definition is that a node's local view of $LH_{b}$ may include blocks that are not in $LH_{b'}$, if they have different watermarks. Therefore, we need to show that for any block $b\in LH_{b'}$, their Sorted Causal History with limited look-back excludes blocks from the same rounds, that is, their watermarks are the same.

\begin{lemma}
    Consider a leader block $b'$ from round $r'$ that eventually commits. For all $b \in LH_{b'}$, the watermark $m_b$ of block $b$ is identical to watermark $m_{b'}$ of block $b'$.
\end{lemma}
\begin{proof}
    Suppose that for any node's local view of the DAG, there exists a block $b\in LH_{b'}$ where its watermark differs from $m_{b'}$. By the definition of watermark of $b$ and $b'$, this implies that there exist two distinct next possibly committed leaders from 2 different rounds, which is absurd. 
\end{proof}

The next question to address is whether for a given block $b$ (from round $r$), the Sorted Causal History with Limited Look-back of $b$ will differ for two distinct nodes. Observe that the only factor that changes the local view of the watermark is the knowledge of the latest committed leader. If a node evaluates a block to have met the criteria for SBO, the criterion will hold regardless of which leader prior to $r$ is newly committed. This consistency is illustrated in \cref{app:lemon proofs}.

In other words, suppose the view of the watermark is for a fixed block $b$ for a node is $m_{b1}$ and it evaluates $b$ to have SBO. Even if the view of the watermark for $b$ is $m_{b2}$ where $m_{b2} > m_{b1}$ for another node, it will not alter the execution prefix of $b$ with respect to the leader that eventually commits it. Intuitively, this holds because the node with the lower watermark must account for a larger set of conflicts.

It is also important to note that even in the case outlined using \cref{def: sorted causal history of a block}, it is not possible for a node to initially evaluate a block as meeting the SBO criterion, and later, after learning of a more recently committed leader, revise this evaluation to not meet the SBO criterion. This holds as well when utilizing \cref{def app: lh}.

Thus, the use of Sorted Causal Histories with Limited Look-back eventually eliminates dangling blocks from consideration. By constantly updating the threshold for the \textit{oldest uncommitted block}, this mechanism refreshes the possibility for blocks to meet the SBO criterion.

\section{Additional Experimental Discussion}
\label{app: experiment discussion}

\subsection{Rationale for Randomizing Faulty Nodes}
There are a few reasons to randomize which nodes are faulty; firstly, in the case where there is a single fault:
In the original Bullshark\cite{bullshark} implementation, the steady leader is elected in a round robin manner. Which means it is possible that the supposed faulty node will never be the 2nd steady leader of a wave, causing the system to never enter fallback-mode. 
Secondly, the original implementation orders the node randomly and then selects the leader in a round-robin manner, enabling the previous case even when $f>1$. 
Artificially improving the performance of the protocol when faults exist.

Lastly, since \sys requires a recursive chain of SBO for a block to gain SBO (if it was not the latest uncommitted block). Randomizing the failures will result in a worse but fairer representation of \sys. 

\subsection{Major Code Differences}
\label{app: sec: code diff}

Here we discuss some changes made to Bullshark's Code. 
\begin{enumerate}
    \item \label{modi:1} We were unable to get the claimed performance when utilizing parameters as described in their paper (utilizing their code). However, we achieved similar results when we utilized the parameters present in their repository\footnote{Data available at \href{https://github.com/asonnino/narwhal/tree/bullshark-fallback/benchmark/data/latest}{GitHub}}. We believe this could either be a mistake on their part or one of the various bugs in their fallback version of the protocol.  
    
    \item Similar to Narwhal-Tusk\cite{narwhal}, nodes perform 1 round of one-to-all broadcast to create batches, and another one-to-all broadcast to form blocks including those batches. 
    Each batch contains up to 500kB of transactions, but can be represented by a 32B hash. Therefore, a block size of 1000B should only contain approximately 32 batches. 
    In their code, whenever producing a block, the entire batch queue is flushed, resulting in blocks containing an unbounded number of batches (and an extremely large block including many hashes). 
    We set an upper-limit (bounded by the block size); therefore, we exhibit a queuing behavior in our results when the client transaction rate is increased. 
    
    \item
    In Bullshark's implementation, they ordered the nodes sequentially and selected the last $f$ nodes to be faulty. 
    We instead, chose to randomized the selection of faulty nodes. 
    Furthermore, we also randomized the steady leader election as compared to a simple round-robin rotation (per Bullshark). 
    However, we add a restriction that no two consecutive steady leaders are the same. This is to introduce a more normalized yet realistic failure behavior. 
    Previously, because the rotation of leaders was static, it was possible for the system to never enter fall-back mode in some cases; when a single faulty node is only ever elected as the first steady leader of the wave--- skewing the results. 
    Our procedure allows the state to enter \textit{fall-back} mode reliably when $f=1$, making it more aligned with an adaptive adversary. 
    
\end{enumerate}

\subsection{Additional Experiment results}
\label{app: add. results}
\paragraph{Varying ``Cross-shard probability'' in \sys}
In \cref{sec: exp type b} we utilized ``Cross-shard probability''= 50\%, denoting the percentage of blocks containing Type $\beta/\gamma$ transactions; in \cref{fig:app:exp} we show the effects of varying that rate. 
Even at 100\%, we enjoy $\sim13\%$ E2E latency reduction and $\sim18\%$ consensus latency reduction. 
\begin{figure}[h]
    \centering
    \includegraphics[width=.8\linewidth]{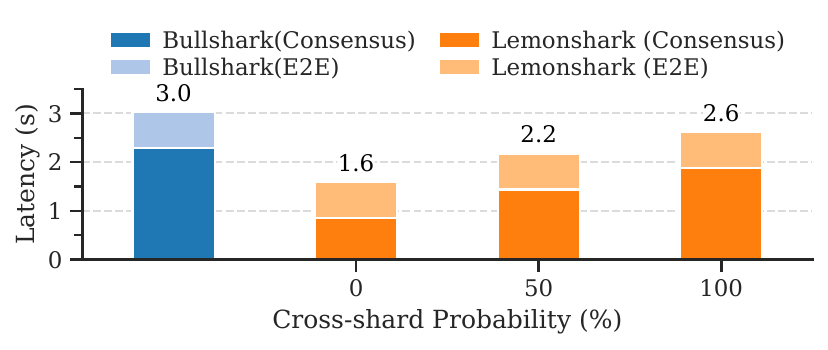}
    \caption{Performance of \sys with moderate amount of Type $\beta$ transactions,(Cross-shard count = 4, Cross-shard Failure = 33\%) while varying Cross-shard probability.}
    \label{fig:app:exp}
\end{figure}

\subsection{Cost of Experiments}
As illustrated in the appendix of Narhwal-Tusk \cite{narwhal}, running these experiments on AWS can become very costly, extremely quickly.

\section{Pipelining Dependent Client Transactions}

\label{app: pipelined}
\begin{figure}[t]
    \centering
    \includegraphics[width=0.8\linewidth,trim=2cm .5cm 1cm 1cm]{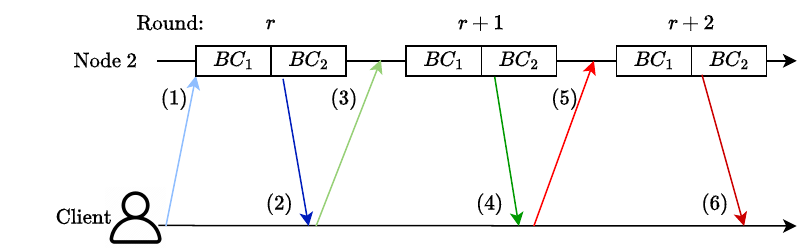}
    \caption{For illustration, we break the RBC of our block into two one/all-to-all broadcast instances $BC_1,BC_2$; akin to the 2-phase nature of Bracha's reliable broadcast\cite{Bracha1987AsynchronousBA}. The client first sends in a transaction $t_1$. (2) After the first $BC_1$, the process provides a speculative outcome for $t_1$ in the form of $to_1$. (3) The client then sends in the next transaction $t_2$ with the requirement that it only executes if the execution of $t_1$ yields $to_1$. (4,5,6) similar steps repeat for $t_2,t_3$. }
    \label{fig: causal example 1}
\end{figure}

Beyond the scope of our main insight in \cref{sec: lemon}, we observe that if a client has a chain of transactions $t_1, \ldots, t_l$ such that every $t_i$ requires the outcome of the previous transaction $t_{i-1}$, a client will typically need to send $t_1$, wait for its finalized outcome when committed, before sending $t_2$, and so on. Therefore, the total latency incurred is at least the time required for $l$ distinct leaders to be committed.

Since a RBC primitive may be decomposed into two phases\cite{Bracha1987AsynchronousBA} of one-to-all broadcast. Suppose the process receiving the transaction provides the (unfinalized and potentially unsafe) speculative TO for $t_1$ in $b_1$ for round $r_1$ during the first phase. In that case, the client may send a tentative transaction $t'_2$ that executes conditionally on the speculated outcome of $t_1$ being the one specified in $t'_2$.

Before the first phase of the \textit{next block} $b_2: t'_2 \in b_2, b_2 \in H_{b_1}$ in round $r_1+1$ can be proposed, a similar operation is performed for $t'_3 \in b_3$, and so forth.

Such pipelined behavior can potentially achieve much lower latencies for the entire chain of transactions $t_1, \ldots, t_l$; we illustrate this in \cref{fig: causal example 1}. It is important to note that due to our sorting algorithm in \cref{def: sorted causal history of a block}, $\forall b: b_1, b_2 \in H_{b}$, $b_2$ is always ordered after $b_1$ in $H_{b}$.

Eventually, when the block is committed or has satisfied the conditions for early finality, the finalized outcome of $t_1$ is communicated to the client, whereupon two cases may occur:

\begin{enumerate}
    \item \textbf{The finalized outcome matches the speculated outcome provided previously:} In this case, the client simply continues, knowing that $t'_2$ has been submitted and will be executed in due time.

    \item \textbf{The finalized outcome does \textit{not} match the previously speculated outcome:} In this case, $t'_2$ will only execute if the speculative outcome of $t_1$ aligns with the one included in $t'_2$. If not, $t'_2$ is aborted, as are any subsequent transactions that depend on its speculated outcome. The client then immediately submits a new transaction, $t''_2$, based on the finalized outcome of $t_1$, thereby restarting the chain of transactions.
\end{enumerate}

Because cascading failure is intrinsic to the transaction chain itself, the above conclusions can be reached independently by each process without requiring coordination. Even if an unexpected outcome occurs, the latency remains upper-bounded by the rate of normal consensus progress via leader election. Thus, in the worst case, this pipelined behavior incurs no additional penalties.

\subsection{Pipelining and \sys}
\begin{figure}[t]
    \centering
    \includegraphics[width=.8\linewidth,trim=2cm .5cm 1cm 1cm]{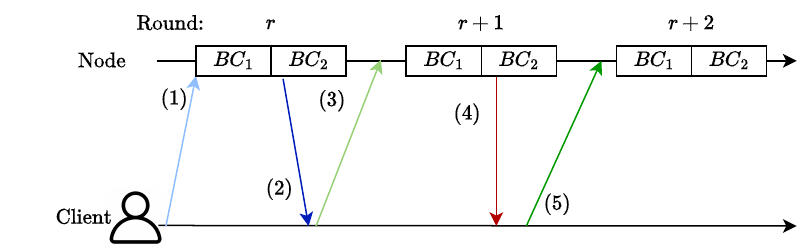}
    \caption{Catching the next bus: (1) Our client sends a Type $\beta$ transaction to the process. (2) After the first $BC_1$, the process sends a speculative outcome. (3) The client uses this speculative outcome to send its 2nd causally related transaction. (4) At $r+1$ we learn that our read will be stale, and the process informs the client (5) The client then re-sends the 2nd transaction in the immediate subsequent block, incurring a single block worth of delay compared to an entire consensus latency worth.}
    \label{fig: bus catch}
\end{figure}

\begin{figure}[!t]
    \centering
    \includegraphics[width=1\linewidth]{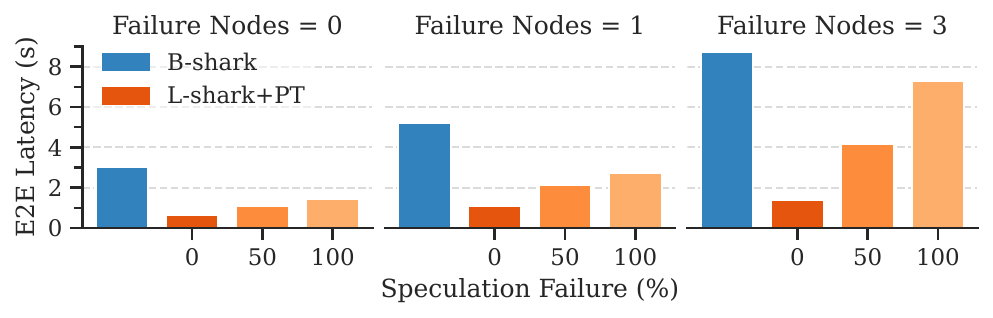}
    \caption{
    Performance of Pipelined client transactions (L-shark + PT) when varying the number of crash-faults. Where $50\%$ of transactions are Type $\beta/\gamma$, ``Cross-shard Count'' = 4, and ''Cross-shard Failure'' = 33\%. 
    }
    \label{fig:exp3-1}
\end{figure}
 
\sys allows processes to determine whether transactions may have STO or whether they definitively cannot have STO. In such cases, the provided speculative TO may be known to be definitively inaccurate before commitment. This can be communicated to the client before finality, thereby allowing the client to resubmit its chain of transactions earlier; we illustrate this in \cref{fig: bus catch}.

However, transactions that may execute based on some condition are analogous to Type $\gamma$ transactions, in that the conditional execution of a transaction may affect other transactions. As such, all contingent transactions are added to the Delay List as well. They are removed when they are finalized or removed entirely when cascading failure occurs.

\subsection{Evaluation}
We also evaluated how pipelined dependent transactions may perform by fixing a moderate amount of Type $\beta$ transactions (Cross-shard count = 4, Cross Shard Failure = 33\%), and varying the number of failures. We also added a parameter ``Speculation Failure'' to denote the probability that the speculated outcome deviates from the finalized outcome if STO is not possible. In the best case, we saw an $\sim80\%$ improvement in latencies when there are no faults, but even in the worst case we still saw $\sim 11\%$ lower E2E latencies. This is because of the phenomenon we describe in \cref{fig: bus catch}.


\end{document}